\documentclass[journal,draftcls,onecolumn,12pt,twoside]{IEEEtranTCOM}
\usepackage{color}     
\usepackage{epsf,psfrag,amssymb,amsfonts,latexsym,cite,verbatim,enumerate,subfigure}
\usepackage{srcltx,amsmath,cases, graphicx}
\usepackage{times, multirow, array}
\usepackage[mathscr]{eucal}
\usepackage{algorithm,algorithmic}

\def\psfancypar#1#2{\begingroup\def\par{\endgraf\endgroup\lineskiplimit=0pt}
               \setbox2=\hbox{\large\sc #2}
               \newdimen\tmpht \tmpht \ht2 \advance\tmpht by \baselineskip
               \font\hhuge=Times-Bold at \tmpht
               \setbox1=\hbox{{\hhuge #1}}
               \count7=\tmpht \count8=\ht1
               \divide\count8 by 1000 \divide\count7 by \count8 
               \tmpht=.001\tmpht\multiply\tmpht by \count7 
               \font\hhuge=Times-Bold at \tmpht
               \setbox1=\hbox{{\hhuge #1}}
               \noindent
                \hangindent1.05\wd1
               \hangafter=-2 {\hskip-\hangindent
               \lower1\ht1\hbox{\raise1.0\ht2\copy1}%
                \kern-0\wd1}\copy2\lineskiplimit=-1000pt}

\newcommand{\E}{\mbox{{\rm E}}}

\def\boxit#1{\vbox{\hrule\hbox{\vrule\kern3pt
        \vbox{\kern3pt#1\kern3pt}\kern3pt\vrule}\hrule}}

\def\reals{ { {\rm  I \kern-0.15em R }  } }
\def\complex{ {\,{{\rm C} \kern-0.50em \raise0.20ex {  |}}\, }}

\def\ubf{{\bf u}}
\def\vbf{{\bf v}}

\def\xbf{{\bf x}}
\def\ybf{{\bf y}}

\def\xbf{{\bf x}}
\def\ybf{{\bf y}}

\def\Rbf{{\bf R}}

\def\Ubf{{\bf U}}
\def\Vbf{{\bf V}}

\def\Xbf{{\bf X}}
\def\Ybf{{\bf Y}}

\def\Cc{{\cal C}}

\def\Ec{{\cal E}}

\def\Gc{{\cal G}}

\def\Nc{{\cal N}}

\def\Rc{{\cal R}}

\def\Uc{{\cal U}}

\def\Xc{{\cal X}}
\def\Yc{{\cal Y}}

\def\be{\vskip .3cm \begin{equation}}
\def\ee{\end{equation} \vskip .4cm \noindent}
\def\defeq{{\stackrel{\Delta}{=}}}
\newcommand{\R}{\mbox{$\hat {\bf R}_{N}$}}

\def\Rxx{\Rbf_{\ssstyle X\kern-.1em X}}

\let\ssstyle=\scriptscriptstyle

\def\Kout{\setbox1=\hbox{\Huge\bf K}\hbox to
1.05\wd1{\hspace{.05\wd1}
}}
\def\Sout{\setbox1=\hbox{\Huge\bf S}\hbox to 1.05\wd1{\hspace{.05\wd1}
}}

  \ifx\LabelFigloaded\MYundefined\relax
  \else
   \fi

  \def\LabelFigloaded{\relax}

  \chardef\LabelFigCatAt\the\catcode`\@
  \catcode`\@=11

 \let\LabelFigwlog@ld\wlog
 \def\wlog#1{\relax}

 \ifx\\\MYundefined@
    \let\\\relax
 \fi

  \def\ms@g{\immediate\write16}

 \def\N@wif{\csname newif\endcsname }
 \def\Temp@ {\N@wif\ifIN@}
 \ifx\INN@\MYundefined@
    \else \let\Temp@\relax
 \fi
 \Temp@

  \def\IN@{\expandafter\INN@\expandafter}
  \long\def\INN@0#1@#2@{\long\def\NI@##1#1##2##3\ENDNI@
    {\ifx\m@rker##2\IN@false\else\IN@true\fi}%
     \expandafter\NI@#2@@#1\m@rker\ENDNI@}
  \def\m@rker{\m@@rker}
 
  \newtoks\Initialtoks@  \newtoks\Terminaltoks@
  \def\SPLIT@{\expandafter\SPLITT@\expandafter}
  \def\SPLITT@0#1@#2@{\def\TTILPS@##1#1##2@{%
     \Initialtoks@{##1}\Terminaltoks@{##2}}\expandafter\TTILPS@#2@}

 \def\Shifted@@#1#2#3{\setbox0=\hbox{#3}%
   \raise -\dp0\vbox {\kern-#2%
       \hbox {\kern#1\unhbox0\kern-#1}%
           \kern#2}}

 \newcount\gridcount
 \newbox\auxGridbox@ \newbox\hGridbox@ \newbox\vGridbox@
 \newbox\Labelbox@ \newbox\auxLabelbox@
 \newbox\Coordinatebox@
 \newtoks\Labeltoks@
 \newdimen\Wdd@ \newdimen\Htt@
 \newdimen\Wddd@ \newdimen\Httt@
 
 \def\Wr@{\immediate\write16}

 \newdimen\GL@wd
 \def\GridLineWidth#1{\GL@wd=#1}

 \def\gobble#1{}
 \def\EdgeErr@{\Wr@{}%
      \Wr@{\string\Edges\space argument
      1, 10, 100 or 1000 please\string!}%
      }

 \newcount\Edgect@

 \def\Sweepup#1\endSweepup{}

 \def\SetEdges@{%
    \edef\Zr@@s{\expandafter\gobble\number\Edgect@\empty}%
        \count255=0\Zr@@s\relax
        \ifnum\count255=\z@\else\EdgeErr@\show\tailtest\fi
        \count255=1\Zr@@s\relax
        \ifnum\count255=\Edgect@\relax\else\EdgeErr@\show\leadtest\fi
    \EdgGl@b\edef\Zr@s{\expandafter\gobble\Zr@@s\empty}
    \ifnum\Edgect@>\@ne\relax\EdgGl@b\let\L@Dc\empty
        \else\EdgGl@b\edef\L@Dc{\string.}\fi
    \ifnum\Edgect@>\@ne\relax
        \EdgGl@b\edef\Edgescale@##1{\divide##1 by \Edgect@}%
        \else\EdgGl@b\edef\Edgescale@##1{}\fi
    }

 \def\Edges#1{\Edgect@=#1\relax
     \let\EdgGl@b\global \SetEdges@}

 \Edges{1}

 \def\hhrule{\hrule height \GL@wd\vskip-.\GL@wd}

 \def\hRule@{%
   \advance\gridcount -2%
   \vfil\hhrule\vfil
   \llap{\smash{\raise -2.5pt
     \hbox{\L@Dc\number\gridcount\Zr@s\kern2pt}}}%
   \hhrule
   }

\def\vvrule{\vrule width \GL@wd \kern-\GL@wd}

 \def\vRule@{\advance\gridcount 2%
   \hfil\vvrule\hfil
   \setbox\auxGridbox@=\vbox to 0pt
      {\vskip \Htt@\vskip 2pt
        \hbox to 0pt{\hss\L@Dc\number\gridcount\Zr@s\hss}\vss}%
      \wd\auxGridbox@=0pt \box\auxGridbox@
   \vvrule
   }

 \def\PlaceGrid@@{\gridcount=10 
  \setbox\hGridbox@=\hbox{%
        \hbox{%
             \hskip-.4pt\vrule
             \vbox to \Htt@{%
               \offinterlineskip\parindent=\z@\relax
               \hbox to \Wdd@{\hfil}
               \hRule@\hRule@\hRule@\hRule@
               \vfil\hhrule\vfil}%
             \vrule\hskip-.4pt}
    }%
  \gridcount=0%
  \setbox\vGridbox@=\hbox{%
      \vbox{\offinterlineskip\parindent=0pt\hsize=0pt
         \vskip-.4pt\hrule%
         \hbox to \Wdd@{%
                 \vtop to \Htt@{\vfil}%
                 \vRule@\vRule@\vRule@\vRule@
                 \hfil\vvrule\hfil}%
         \hrule\vskip-.4pt}}%
  \wd\hGridbox@=0pt\ht\hGridbox@=0pt
  \wd\vGridbox@=0pt\ht\vGridbox@=0pt
  \hbox{\box\hGridbox@\box\vGridbox@}%
  }

 \def\LabelsGlobal{\def\LabGl@b{\global}}
 \def\LabelsLocal{\def\LabGl@b{}}
 \LabelsGlobal 

 \def\SetLabels#1\endSetLabels{%
   \LabGl@b\Labeltoks@={#1()\\}%
   }

 \LabGl@b\Labeltoks@={()\\}

 \def\ShowGrid{\LabGl@b\let\PlaceGrid@\PlaceGrid@@}
 \def\HideGrid{\LabGl@b\let\PlaceGrid@\relax}
 \def\Grids{\ShowGrid\LabGl@b\let\GridSwitch@\ShowGrid}
 \def\noGrids{\HideGrid\LabGl@b\let\GridSwitch@\HideGrid}

 \noGrids

 \def\bAdjust@@{%
     \setbox\auxLabelbox@=\hbox{\raise \dp\auxLabelbox@
            \box\auxLabelbox@}}
 \def\bAdjust@{\let\vAdjust@\bAdjust@@}

 \def\eAdjust@@{\dimen0=-.5\ht\auxLabelbox@
     \advance\dimen0 by .5\dp\auxLabelbox@
     \setbox\auxLabelbox@=
            \hbox{\raise\dimen0\box\auxLabelbox@}}
 \def\eAdjust@{\let\vAdjust@\eAdjust@@}

 \def\tAdjust@@{%
     \setbox\auxLabelbox@=\hbox{\raise-\ht\auxLabelbox@
            \box\auxLabelbox@}}
 \def\tAdjust@{\let\vAdjust@\tAdjust@@}

 \let\vAdjust@\relax

 \def\lAdjust@{\let\hAdjust@\rlap}
 \def\rAdjust@{\let\hAdjust@\llap}

 \let\hAdjust@\relax\let\vAdjust@\relax

 \def\FetchLabel@#1(#2)#3\\{%
     \IN@0#2@@\ifIN@
        \setbox0=\hbox{\ignorespaces#1#3\unskip}%
        \ifdim\wd0>0pt
           \ms@g{}%
           \ms@g{ !!! Bad label(s)? !!!}%
           \message{ #1(#2)#3}%
        \fi
        \def\LabelMole@##1\endFetchLabel@{%
            \IN@0()\\@##1@%
            \ifIN@\def\Temp@{\FetchLabel@##1\endFetchLabel@}%
            \else\def\Temp@{}%
            \fi
            \Temp@
           }%
     \else
       \ignorespaces#1\unskip
       \setbox\auxLabelbox@=%
         \hbox to 0pt{\hss\ignorespaces\hAdjust@
          {\ignorespaces#3\unskip}\hss}%
       \vAdjust@
       \let\hAdjust@\relax\let\vAdjust@\relax
       \AugmentLabelBox@@{#2}%
       \ht\Labelbox@=0pt\dp\Labelbox@=0pt
       \let\LabelMole@\FetchLabel@%
     \fi\LabelMole@}

 \newtoks\XYSep@ 
 \def\SetXYSeparator#1{%
     \IN@0#1@@\ifIN@\XYSep@{*}%
     \else
     \XYSep@{#1}%
     \fi
     }

 \SetXYSeparator*

 \def\AugmentLabelBox@@#1{%
     \IN@0\the\XYSep@ @#1@\ifIN@
       \SPLIT@0\the\XYSep@ @#1@%
       \setbox\Labelbox@=\hbox to 0pt{%
         \unhbox\Labelbox@
         \Shifted@@{\the\Initialtoks@\Wddd@}%
         {\the\Terminaltoks@\Httt@}%
         {\box\auxLabelbox@}}%
     \else
         \ms@g{}%
         \ms@g{ !!! Bad insertion point. !!!}%
         \message{ (#1\ this point was rejected.)}%
     \fi
    }

 \def\FetchOption@#1[#2]#3\endFetchOption@{%
    \def\temp{#1}
    \ifx\temp\empty
       \Edgect@=#2\relax
       \let\EdgGl@b\relax
       \SetEdges@
       \Cleaner@#3%
    \fi}

 \def\Cleaner@#1[@]{\Labeltoks@{#1}}
     
 \def\PlaceLabels@@{\mathsurround=0pt
     \def\Cr@{\\}%
     \let\L\lAdjust@\let\R\rAdjust@
     \let\B\bAdjust@\let\E\eAdjust@\let\T\tAdjust@
     \expandafter\FetchOption@\the\Labeltoks@[@]\endFetchOption@
     \Wddd@=\Wdd@ \Edgescale@\Wddd@ 
     \Httt@=\Htt@ \Edgescale@\Httt@
     \expandafter\FetchLabel@\the\Labeltoks@\endFetchLabel@
     \box\Labelbox@
     }%

 \let \PlaceLabels@\PlaceLabels@@

 \def\AffixLabels#1{\setbox\Coordinatebox@=\hbox{#1}%
      \Wdd@=\wd\Coordinatebox@ \Htt@=\ht\Coordinatebox@
      \advance\Htt@ \dp\Coordinatebox@
      \hbox{\copy\Coordinatebox@\kern-\Wdd@ 
           \Shifted@@{0pt}{-\dp\Coordinatebox@}%
           {\PlaceLabels@\PlaceGrid@}%
           \kern\Wdd@}%
      \GridSwitch@ 
      \LabGl@b\Labeltoks@{()\\}%
      }
 
   \let\wlog\LabelFigwlog@ld   
   \catcode`\@=\LabelFigCatAt  

                                By

              Raymond S\'eroul <A18645@FRCCSC21.BITNET>
                                and 
              Laurent Siebenmann <lcs@topo.math.u-psud.fr>
    
              VERSIONS: July 1991, Oct 1991, Jan 1992, July 1992

INTRODUCTION

      This labelling package is intended for TeX users who
rely on non-TeX sources for for their graphics inserts.  It
provides means for adding TeX labels to such inserts with a
minimum of fuss. 

       For most labels, TeX users have in the past found it
reasonably convenient to rely on non-TeX sources. Typical
occasions when an inescapable need for TeX labels seemed to
arise are

 (a) when the graphics program lacks certain exotic or complex
mathematical symbols

 (b) when the very highest typographical quality is wanted for the
labels

 (c) when labels included with the graphics fail to print, 
 and you cannot figure out why (cf. boxedeps.doc).  The labels
 provided by labelfig.tex are 100

       Since this package first appeared, many users, who in the
past scarcely dreamed of using TeX labels, have come to use
nothing but.  So it is now appropriate to add

Intoxication Warning:  TeX labels may be addictive and expensive. 

     If you have a fast preview you may disagree, and even find
that this package provides an agreeable paste-up environment; see
extra applications at end.

     Note to publishers: It is possible and convenient to ultimately
export the TeX labels produced by labelfig.tex to become an integral
part of the EPS file. This is often desired by a publisher who typically
uses an "upmarket" graphics or page layout program, with which the
staff is skilled in perfecting figures.  See Appendix I for
a recipe.

     The authors are grateful to Patrick Ion of Math Reviews for
helpful comments and encouragement.

BASIC INSTRUCTIONS

    After reading in the macro file using

preview or proof your figure with a coordinate grid printed on
top, by typing the following:

    \ShowGrid  
    \AffixLabels{<the graphics insertion>}

Here <the graphics insertion> is what you would type to insert
the graphics object alone without the grid.  This must provide
for the space around it. For example <the graphics insertion>
might well be \BoxedEPSF{MyFigure scaled 700} using the
boxedeps.tex macro package (from same source); this provides a
TeX box containing the encapsulated PostScript insert specified by
the file MyFigure. \AffixLabels{...} provides the grid (supposing
\ShowGrid is present) and later, once you have specified labels
using the grid, it will "tack on" the labels.

     The grid is a sort of (usually elongated) checkerboard of
ten rows and ten columns and its (internal) partitions are by
default numbered  .1, ... ,.9  both horizontally (X-coordinate
running left to right) and vertically (Y-coordinate running bottom
to top).  Thus the points enclosed by the grid correspond to the
points of the unit square in the cartesian "X-Y" plane, the lower
left corner corresponding to the origin (0,0).  By extrapolation,
the full page corresponds to a larger rectangle in the plane.

     These coordinates serve to position labels as follows.
Before the \AffixLabels{...} command type label specifications:

  \SetLabels
   (<X-coordinate>*<Y-coordinate>) <first label> \\
   .
   .
   .
   (<X-coordinate>*<Y-coordinate>)  <last label> \\
  \endSetLabels

Each row specifies one label and is terminated by \\.  In each
row, the position indicator comes first; it is written as a
standard cartesian point except that the X- and Y- coordinates
are separated by * rather than a comma because TeX allows a
comma as decimal point. There are no dimension units to specify
as the unit is the grid itself.

     By default, this cartesian point specifies where the middle
of the baseline of the label will be located.  However if you precede
the point by \L [or \R] the left [or right] edge of the baseline will
be located there. Similarly you may also precede the point by \T, \E,
or \B to vertically align the top equator or bottom of the label box
at the specified point.  This gives nine standard positions of
the label with respect to the insertion point --- corresponding to
the eight principle points of the compas and the center

                     \L\T     \T      \R\T

                     \L\E     \E      \R\E

                     \L\B     \B      \R\B

But this neglects the default "baseline" level of TeX,
giving potentially three more positions

                     \L    <no tag>   \R

For text, the baseline level is often the preferred. Its relation to
the others is variable. It will often coincide with the bottom level,
as happens for "X".  But it is often distinct, as for "g", in which
case you have in all 12 distinct positions rather than 9.

     It is convenient to think of this specification of label
position as attaching the label by a thumb-tack to the coordinate
grid. There are up to twelve positions of the thumb-tack on the
label, while the position of the thumb-tack on the coordinate grid is
arbitrary.  Normally, one choses the position of the thumb-tack on
the label to be the one that is the closest to the item being
labeled.  There are good reasons for this "rule of thumb":

   (a)  It facilitates correct positioning at first try.

   (b)  If the scale of the figure must be altered after labels
have been affixed, the labels have a good chance of remaining well
positioned.

   (c)  The visible grid need not extend beyond the "bounding box"
for the figure, because the best preferred position is always
(at least almost) within the bounding box .

The second reason is particularly important. Indeed it often
happens that scale has to be altered after labelling begins, in
order to either provide space for the labels, or to adjust
proportions between the labels and the figure.  (The size of labels
is unaffected by scaling.)

     Here is an artificial but self-contained test which uses
TeX rules to make a graphics object.

TEST

    Do not skip this!

 \def\FrameIt#1{\hbox{\vrule$\vcenter {\hrule\kern3pt%
             \hbox {\kern3pt #1\kern3pt}%
               \kern3pt\hrule}$\relax\vrule}}

 \def\Caption#1#2{\FrameIt{%
       \vtop {\hsize=#1\relax \parindent=0pt
         \leftskip=0pt \rightskip=0pt plus15pt
         \parfillskip=0pt
         \lineskip=1pt\baselineskip=0pt
         #2}}}

 \def\FirstQuadrant{\hbox to 100pt{\vrule\vbox to 100pt{%
        \hbox to 100pt{\hfil}\vfil\hrule}\hss}}

  \SetLabels
    \R(.5*.2) $\zeta\,\cdot$\\
    (.9*-.10) $\xi$\\
    \R(-.03*.9) $\eta$\\
    \T(.5*.9) \Caption{70pt}{%
          \it The norm of
          $g(\xi+i\eta)$ is indicated on
          contours of this invisible surface.}\\
  \endSetLabels

  \AffixLabels{\FirstQuadrant}

  \end

  Note that the coordinates to use for labels are indicated on the
edges of the grid (when visible) corresponding to the conventional
x- and y- axes of the Cartesian plane. By default the grid is
1-by-1. However, by the command \Edges{100}, you can change this
to 100-by-100 and many users find this alternative most
convenient. Place the command \Edges{...} in your style file (or
header) since its effect is is global. Other possible edge values
are 10 and 1000.

  If you use the command \Edges{...} at all, do so with care.  For
if you accidentally delete an \Edges{...} command your labels will
abruptly be badly misplaced and may logically but mysteriously
generate "dimension too big" errors under TeX and "off page" errors
under your driver.  

  You can dictate the edgescale for an individual figure by giving
the scale in brackets immediately after \SetLabels.  Thus, to
import into an article using say \Edge{100} a figure labelled using
another edgescale, say the original 1-by-1 default, you can use
\SetLabels[1]...\endSetLabels.

GETTING IT DOWN PAT

     Complicated labeling deserves the same respect as
complicated mathematics.  Do not expect it to come out perfect the
first time!  What is needed in either case is a mechanism to
repeatedly typeset troublesome pieces.

     One mechanism is always available.  One does complicated
labelling in a separate "test" file involving just the figure being
labelled;  a texpert will know how to \dump TeX's current state as
a temporary format that restarts rapidly at each retry.  Usually,
one then pastes the completed labelled figure back into the main
TeX file, but, of course, one can also \input it as an auxiliary
file.

     If you do not have a TeXpert at handy, here is a first
approximation to an efficient setup. By deletions reduce a copy
of your article to just a few lines before and after the figure.
Now label the figure, and finally, copy and paste the labelled
figure to the original article. Then copy the next figure to label
into this testbed and repeat. The TeXpert can improve the  speed
at which TeX starts up, by compiling a format specifically for
your article; just one caution: best NOT include in the format
ephemeral details of setup like \Set<mydriver>ArtSpecials (from
boxedeps.tex because this reads  figure dimensions which you may
change during your work session.

     An improved mechanism to repeatedly typeset troublesome
pieces is now available on the Macintosh; it is called LinoTeX;
see the same ftp sources.  It could be set up on many types
of computer.

     Before using labelfig.tex to attach labels to a graphics
object inserted using boxedeps.tex or BoxedArt.tex, make it a
firm rule to carefully adjust the bounding box using the trimming
commands of these packages, and also at least tentatively scale
and position the object. Beware of changing the grid inadvertently
after the labels have been positioned.  For example, correcting
the bounding box of a PostScript graphics object can foul up the
labels by changing the coordinate grid to which the labels are
attached. This is particularly true for the trimming  commands of
boxedeps.tex and BoxedArt.tex. However, as noted already, change
of scale is much less disruptive, and modest adjustments should be
well tolerated.

     Sometimes the labels protrude so far from the bounding box
of a figure that the figure has to be repositioned.  Best do this
by ad hoc spacing, say using \hglue and \vglue; altering the
bounding box would create a vicious circle.

     Remember that you are responsible for preventing labels
from overlapping. You are responsible for all label typography
including size and style. A label is really just about anything
that can be put in a TeX box. Note that spaces at the beginning
and end of labels will normally be suppressed; if you really want
them you must protect them with TeX braces.

     This package temporarily sets the \mathsurround parameter
of TeX to zero  while the labels are being affixed. This is done
because nonzero \mathsurround space would influence the position
of left and right aligned labels; then, when a texpert or printer
modifies mathsurround, diagram labeling might be disastrously
altered. There is a small price to pay involving labels that are
formatted as caption boxes including mathematics: you  may want or
need to specify an explicit mathsurround space within the caption
box; it will not influence anything outside.

     Those hostile to the use of * as separator between
the X and Y coordinates of label insertion points, are free to
impose another using \SetXYSeparator{<the new separator>}.  
Americans may prefer "," to "*" since they never use a 
comma as a decimal point; on the other hand, * may be more visible.

APPENDIX (I)  MERGING labelfig.tex LABELS INTO AN EPSF GRAPHICS OBJECT.

     As promised in the introduction, here is a recipe useful for
publishers. It works at least on Macintosh and at least for vectorized
graphics and Adobe type1 fonts.  (There is surely a similar recipe for
PCs under MSWindows.)

 (a)  Use boxedeps.tex utility to integrate the figure given by the eps
file, "x.eps" say, with a visible frame around it.  See
\ShowDisplacementBoxes command in boxedeps.tex.  To get precise results
automatically it is important to use the \Trim... commands of
boxedeps.tex making the "DisplacementBox" neatly fit the figure.

 (b)  Use the TeX printer driver and LaserWriter (versions >= 8.1.1) to
export to an EPSF the DVI page containing the integrated, labelled
figure. You now have an EPS file  "xx.eps"  that contains too much, and at
the wrong scale, and at wrong position.

 (c)  Convert the EPSF to an Adode Illustrator format EPSF using
the shareware utility called epsConvert by Sam Weiss
1993-- (currently $25).

 (d)  In Illustrator (or a compatible program), group the labels and the
"DisplacementBox"; copy them to the clipboard and paste them into "x.ps".
This step requires that all the label fonts be "visible to the Macintosh.

 (e)  Translate and scale the pasted group consisting of the labels plus
the "DisplacementBox" so as to make the "DisplacementBox" the bounding
box of (labelless) figure represented by "x.eps".  At this point the
labels will be correctly placed on the figure "x.eps".

 (f)  Ungroup and delete the "DisplacementBox".  The result is the
desired single EPS file, "x+.eps" say, It contains the original figure
plus its labels.  

     Using grouping and ungrouping appropriately in "x+.eps", a
publisher's staff can very efficiently improve label positions etc.

APPENDIX II)  SOME EXOTIC APPLICATIONS

     The grid of labelfig.tex is analogous to a light-table in
classical page makeup with wax or latex glue.  In principle, you
can use it to compose any page from its indivisible parts.  This
even has some of the artisanal charm of classical paste-up
provided you have a fast screen preview to make the process
"interactive".

     In practice labelfig.tex is a tool for nonstandard jobs.
Here are a few going beyond the labelling already discussed.

(I)  GRAPHICS INTEGRATION.

     This is accomplished by treating the imported graphics
objects as labels.  The underlying graphics object is then
typically an empty  \vbox to <dimension>{\vfill} in a TeX
\midinsert...\endinsert construction.  A label line
might be of the form

   (.1*.1) \\

The exact form of the special command varies from driver to
driver.  However, in the case of encapsulated PostScript graphics
(EPSF norm), by relying on boxedeps.tex, one can have the
following standard syntax (independant of driver  (see
boxedeps.doc for details.
  
  (.1*.1) \BoxedEPSF{MyFigure scaled <scale in mils>}\\

This may be slow since it requires TeX to read the PostScript
file to read bounding box using many complex macros.  So you
may want to try

  (.1*.1) \EPSFSpecial{MyFigure}{<scale in mils>}\\

which is fast and driver independant, but it squashes the
bounding box, normally to its lower left corner.

     Similarly for graphics of the Macintosh PICT norm ---
using BoxedArt.tex (same sources) in place of boxedeps.tex.

     This approach to integration is to be recommended when
one is assembling a composite graphics object.

 (II)  COMMUTATIVE DIAGRAM ENHANCEMENT

     Commutative diagrams or arrays of mathematical objects
connected by arrows of various sorts are common in mathematics.
The mathematical objects require the use of TeX.  Recently TeX
acquired a good collection of arrows of all slopes --- that of
LamSTeX --- plus pwerful macros to build the diagrams.

     However, even the LamSTeX collection is often
inadequate; it lacks for example double shafted arrows, dotted
arrows and curved arrows. Fortunately it is possible to produce
such arrows on an individual basis using sophisticated graphics
programs such as Illustrator and AldusFreehand (both serving
the EPSF norm) or using Metafont (with its public domain norm).
Since the creation of each new arrow is a work of love, you
probably want to limit the number of arrows by using LamSTeX
for most arrows. The 40K commutative diagram module of LamSTeX
has been adapted to work with AmSTeX and a copy may be posted
with LabelFig and related files. Unfortunately no one has yet
offered a version that works with Plain TeX or LaTeX.

       Suffice it here to say that when the exotic arrow has
been somehow imported into TeX, labelfig.tex treats it as a
label that one affixes to the commutative diagram.  Two other
steps will be treated in separate notes, namely the matter of
extracting the dimension specifications for the arrow and the
construction of the arrow --- for these steps are far from
unique and often depend intimately on your computer environment. 
Notes for the Macintosh-Textures-Illustrator combination are
found in the file ExoticArrows.doc.

 (III) NESTING 

Ingenuity pays off in exploiting labelfig.tex. One can
mix graphics and typography quite freely.  labelfig.tex is good
for freeform or overlapping arrangements, while boxedeps.tex (or
BoxedArt.tex) is best for regimented non-overlapping
arrangements --- and the two can be combined.

     The default behavior of labelfig.tex is not ideal 
for nesting objects, because to prevent trouble for beginners
the register for labels is globally cleared when \AffixLabels
concludes.  But there are switches available

      \LabelsGlobal      \LabelsLocal

which change this.  To understand this, extend the above test 
by something like:

 \LabelsLocal

 \SetLabels
    (.5*.5) AAA\\
 \endSetLabels

 {
 \SetLabels
    (.5*.5) ZZZ\\
 \endSetLabels
   \AffixLabels{\FirstQuadrant}
 }

   \AffixLabels{\FirstQuadrant}

     There are however potential pitfalls.  Neither
labelfig.tex nor boxedeps.tex has been tested under extreme
conditions. Problems may occur if their procedures are
indiscriminately nested. For boxedeps.tex (not labelfig.tex)
there is a precise cause for worry, namely many of its
variables are "global", which means that TeX braces will not
provide the protection one might expect.

COMMAND SUMMARY FOR labelfig.tex

  Here [...] means optional (one or zero)
       [...]* means any number of such constructs

  \SetLabels
    [[<P>](<X><Sep><Y>) <label> \\]*
  \endSetLabels
  \ShowGrid  
  \AffixLabels{<the figure>}

   --- <P> is tack position, one of eleven or empty
              order irrelevant

                   \L\T      \T      \R\T

                   \L\E      \E      \R\E

                     \L               \R

                   \L\B      \B      \R\B

   --- (<X><Sep><Y>) insertion point;
  <Sep> is separator, = * by default;
  \SetXYSeparator{<Sep>} changes it.
   <X> and <Y> are real numbers

  --- <label> a label to attach 

  --- <the figure> the figure to label 

  \GlobalLabels (default)     
  \LocalLabels  setting for nested constructs.

 \Grids makes ALL grids appear; \HideGrid then makes just next disappear.
 \noGrids returns to default.  The commands are always global.

 \GridLineWidth{<dimension>} adjusts width of grid lines. Default is very
small, to give "hairline" effect. If your grid lines are missing try
setting \GridLineWidth{1pt}.

 \Edges#1 globally changes the edge size of all grids to the numerical 
value #1, which must be 1, 10, 100, or 1000.  The default is 1.

VERSION HISTORY.
 --- Jan 1993: \Edges#1 and [??] option after \SetLabels
 --- July 1992: \Grids, \noGrids, \HideGrid;
       Gridlines become hairlines; \GridLineWidth{<dimension>}.
 --- Oct 1991, Jan 1992: \SetXYSeparator{<Sep>},  \LabelsGlobal,
       \LabelsLocal.
 --- July 1991: first release

Address for bugs and other feedback:

        Raymond S\'eroul
        IREM and Lab. de Typographie Informatise
        Univ. Rene Descartes
        Strasbourg

    Tel 33-88-41-63-45
    Email:  A18645@FRCCSC21.BITNET

        Laurent Siebenmann
        Mathematique, Bat. 425,
        Univ de Paris-Sud,
        91405-Orsay,
        France

    Tel 33-1-6941-7949; 
    Email: lcs@topo.math.u-psud.fr

\def\scalefig#1{\epsfxsize #1\textwidth}
\def\defeq{\stackrel{\Delta}{=}}

\def\tcr{\textcolor{red}}
\def\tcb{\textcolor{blue}}

\def\tcr{\textcolor{red}}
\def\tcb{\textcolor{blue}}

\def\bf{\textbf}

\newcommand {\Ebb}{{\mathbb{E}}}

\newtheorem{theorem}{Theorem}

\normalsize

\begin{document}

\title{ \LARGE Enhancing Non-Orthogonal Multiple Access By Forming  Relaying Broadcast Channels}

\author{  Jungho So, {\em Student~Members, IEEE}, and Youngchul
Sung$^\dagger$\thanks{$^\dagger$Corresponding author}, {\em
Senior~Member, IEEE}  \\
\thanks{The authors are with Dept. of Electrical Engineering,  KAIST, Daejeon 305-701, South
Korea. E-mail:\{jhso, ycsung\}@kaist.ac.kr.
This research was supported in part by Basic Science Research Program through the National Research Foundation of Korea (NRF) funded by the Ministry of Education (2013R1A1A2A10060852).
}
}

\markboth{Submitted to IEEE Transactions on Communications, \today}%
{}

\maketitle

\begin{abstract}
In this paper, using relaying broadcast channels (RBCs) as
component channels for non-orthogonal multiple access (NOMA) is
proposed to enhance the performance of NOMA in single-input
single-output (SISO) cellular downlink systems. To analyze the
performance of the proposed scheme, an achievable rate region of a
RBC with compress-and-forward (CF) relaying is newly derived based
on the recent work of noisy network coding (NNC).  Based on the
analysis of the achievable rate region of a RBC with
decode-and-forward (DF) relaying, CF relaying, or CF relaying with
dirty-paper coding (DPC) at the transmitter,  the overall system
performance of NOMA equipped with RBC component channels is
investigated.  It is shown that NOMA with RBC-DF yields marginal
gain and NOMA with RBC-CF/DPC yields drastic gain over the simple
NOMA based on broadcast component channels in a practical system
setup.
 By going beyond simple broadcast channel (BC)/successive interference cancellation (SIC) to
advanced multi-terminal encoding including DPC and CF/NNC, far
larger gains can be obtained for NOMA.
\end{abstract}


\IEEEpeerreviewmaketitle

\section{Introduction}

{\em Motivation:} To meet the exponentially growing demand for
high date rates in next generation wireless communication systems,
enhancing existing lower band wireless systems as well as
introducing new  bandwidths in higher bands is under vigorous
efforts \cite{Li&Niu&Papathanassiou&Wu:14VT}. One of the
technologies for increasing the spectral efficiency of cellular
systems is recently proposed non-orthogonal multiple access (NOMA)
\cite{Saito&Kishiyama&Benjebbour&Nakamura&Li&Higuchi:13VTC,
Ding&Peng&Poor:15ComLetter}. Traditionally, the wireless communication
resources in cellular systems such as time and frequency bandwidth
were divided into orthogonal resource blocks, and within a
separated resource block only one user is served by the base
station (BS).  In NOMA, however, multiple users are served
non-orthogonally within each resource block by exploiting the
power domain. From the system perspective, such user allocation
can be regarded as {\em system overloading} with which the number
of served users is larger than that of orthogonal resource blocks.
Since multiple users are served non-orthogonally within each
resource block with such overloading,  the signal from some users
allocated to a resource block interferes with other users
allocated to  the same resource block, but such interference is
eliminated by partial user cooperation and non-linear decoding
like successive interference cancellation (SIC). For the example
of two user allocation in the same resource block, two users with
different channel gains are grouped into a resource block so that
one user has a higher channel gain (i.e., is close to the BS) and
the other user has a lower channel gain (i.e. is far from the BS).
Then, the signals of the two users are added and transmitted. At
the receiver side, the user close to the BS decodes not only its
data but also the data for the  user far from the BS, and cancels
the signal of the user far from the BS from its received signal.
On the other hand, the user far from the BS just decodes its data
by treating the interference from the user close to the BS as
noise. This is possible due to the asymmetry of the channel gains
of the two users since with power control the user close to the BS
requires less power than the user far from the BS and thus the
user close to the BS can decode the data intended for the user far
from the BS if the user far from the BS can decode its own data.
It has been shown  that such system overloading based on NOMA can
yield non-trivial spectral efficiency increase
\cite{Saito&Kishiyama&Benjebbour&Nakamura&Li&Higuchi:13VTC}.

From the perspective of information theory,  the BS and the unique
served user within a separated resource block form a
point-to-point (P2P) channel in conventional
orthogonalization-based cellular systems. However, under NOMA  a broadcast channel
(BC) is formed by the BS and the users allocated to the same
resource block within a separated resource block. Indeed, it is
known that a single-input single-output (SISO) Gaussian BC (GBC) is a degraded BC and the
aforementioned  super-position coding and SIC achieve its capacity
region \cite{Cover&Thomas:91,Gamal&Kim:11NIT}.  Thus, the rate increase by NOMA is
due to the change of the channel within each resource block from a P2P channel to a BC
since the capacity of a given channel does not change.  The
penalty for the rate increase is the required cooperation between
the served users and the increase in the transmitter and/or
receiver side processing.

Some modification has been made to enhance the aforementioned simple NOMA
 by increasing the level of the cooperation between the
served users  and changing the type of channel within a resource
block \cite{Ding&Peng&Poor:15ComLetter} with the consideration of the
recently available device-to-device (D2D) communication capability
\cite{Doppler&Rinne&Wijting&Ribeiro&Hugl:09ComMag,Fodor&Dahlman&Mildh&Parkvall&Reider&Miklos&Turanyi:12ComMag}.
In \cite{Ding&Peng&Poor:15ComLetter}, the authors considered a
two-phase (half-duplex) cooperative NOMA in which the BS
broadcasts data to both users in the first phase, and the user
with good channel helps the other user by transmitting the data
for the other user decoded at its site in the first phase to the
other user in the second phase. That is, the user with good
channel serves as a half-duplex decode-and-forward (DF)
relay.\footnote{In the case of more than two users in a resource
block, the same idea can be extended to a multi-phase cooperative
NOMA \cite{Ding&Peng&Poor:15ComLetter}.}
  However, such half-duplexing reduces the data rate by half and the resulting system has limitation to increase the system rate.

{\em Summary of Results:} In this paper, we further enhance the
performance of NOMA by introducing full-duplex relaying at the
user with good channel and several relevant encoding schemes at
the BS  for the case of two-user allocation within each resource
block which seems most practical with consideration of performance
gain and complexity.
  When the user with good channel serves as a full-duplex relay, the BS and the two served users form a
 {\em relaying broadcast channel (RBC)} from the perspective of information theory \cite{Liang&Veeravalli:04ISIT,Bross:09IT}.
 A RBC is different from a BC in that one of the receivers serves as a relay as well as a receiver for its own data,
as shown in Fig. \ref{fig:RBCchannel} in Section \ref{sec:componentchannel}.   There exist several known
relaying methods such as amplify-and-forward (AF), DF, and
compress-and-forward (CF)
\cite{Cover&Gamal:79IT,Kramer&Gastpar&Gupta:05IT}. In this paper,
we consider DF and CF relaying for performance improvement. AF is
not relevant in RBCs for NOMA since AF in RBCs amplifies the
signal intended to the relaying receiver as well as the signal
intended for the other receiver and directly transmits the
amplified sum to the other receiver.
 Several information-theoretical achievable rate region analyses were performed on RBCs.
  In \cite{Liang&Veeravalli:04ISIT}, the authors studied the achievable rate region of a RBC with a DF relaying receiver and showed
  that the achievable rate region of a RBC subsumes that of the BC generated by eliminating  the link between the relaying receiver and the other  receiver in the SISO Gaussian case.
In \cite{Bross:09IT}, the author considered the achievable rate
region of a RBC  employing CF with common information based on
\cite{Cover&Gamal:79IT}. However, the encoding scheme at the BS
proposed in \cite{Bross:09IT} is complicated and does not provide
much insight.  Furthermore, we are not much interested in the case
with common information for both receivers. Thus, we here simplify
the problem by eliminating the common information and derive an
achievable rate region of a RBC employing CF based on the recent
work of noisy network coding (NNC)\footnote{Noisy network coding
for the case of three nodes composed of a transmitter, a relay and
a receiver can be viewed as a simplified CF scheme.} in
\cite{Lim&Kim&Gamal&Chung:11IT}. Note that the setup of RBC and
that of NNC are different in that a transmitter in NNC has only
one message possibly intended for many receivers but in RBC the
transmitter has two messages intended for two different users.
Although the channel setup is different, we still apply the NNC
encoding scheme with some modification appropriate to RBC and
derive an achievable rate region of a RBC with CF/NNC.
Furthermore, based on this result, we derive an achievable rate
region of a RBC with CF/NNC when dirty-paper coding (DPC) \cite{Costa:83IT} is
applied at the transmitter.

To evaluate the overall system performance of the proposed NOMA
with RBC, we consider two  user pairing and scheduling methods:
near-far pairing and nearest-neighbor pairing. These two pairing
methods are opposite to each other and provide two extreme pairing
on which the performance of different NOMA schemes can be
compared.  Based on the achievable rate region result for a RBC
with DF in \cite{Liang&Veeravalli:04ISIT} and the newly derived
achievable rate region result for a RBC with CF/NNC  or CF/NNC/DPC
in this paper for each resource block, the overall system
performance gain of the proposed NOMA with RBC is examined under
the two user pairing and scheduling methods. Numerical results
show that the gain of  NOMA with  RBC-DF is marginal, but NOMA
with RBC-CF/NNC/DPC yields drastic gain over the simple NOMA based
on GBC/SIC
\cite{Saito&Kishiyama&Benjebbour&Nakamura&Li&Higuchi:13VTC} in a
practical system setup.

{\em Notations and Organization:} We will make use of standard
notational conventions. Vectors are written in boldface  in
lowercase letters.  Random variables are written in capitals and
the realizations of random variables are written in  lowercase
letters. For a random variable $X$, $\Ebb\{X\}$ denotes the
expectation of $X$, and $X\sim\mathcal{CN}(\mu,\Sigma)$ means that
$X$ is circularly-symmetric complex Gaussian-distributed with mean
$\mu$ and covariance $\Sigma$.

The remainder of this paper is organized as follows. In Section
\ref{sec:systemmodel}, the system model is described. In Section
\ref{sec:componentchannel}, the achievable rate region of a RBC is
given in general discrete memoryless channel and Gaussian channel
cases. The considered user pairing and scheduling methods are
described in Section \ref{sec:pairing}. Numerical results are
provided  in Section \ref{sec:numericalresults}, followed by
conclusion in Section \ref{sec:conclusion}.

\section{System Model}
\label{sec:systemmodel}

In this paper, we consider  a single-cell SISO downlink system
with a single-antenna BS and $K$ single-antenna users, where the
considered cell topology is a typical $120^o$ sector of a disk and
each user is distributed uniformly in the sector, as shown in Fig.
\ref{fig:cellStruc}.  We assume that we have $B$ communication
resource blocks that are orthogonal to each other.\footnote{For
example, such resource orthogonalization can be attained by OFDM
or other orthogonalization techniques. In the case of OFDM, one
resource block represents a subcarrier or a chunk of subcarriers.}
The BS selects and assigns $M$ users to each resource
block\footnote{The scheduling and grouping method will be
explained in Section \ref{sec:pairing}.}  and we assume that $K
\ge BM$ to incorporate the impact of multi-user diversity in our
system performance investigation. In particular, we focus on the
case of $M=2$ in this paper.
\begin{figure}[!hbtp]
\centering
\includegraphics[width=9cm]{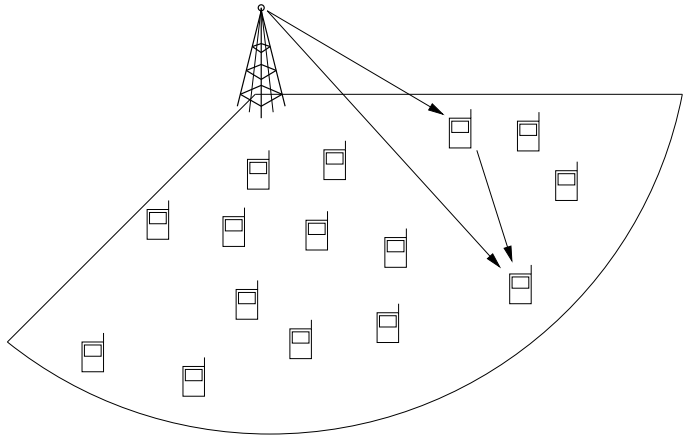}
\caption{The considered single-cell SISO downlink system}
\label{fig:cellStruc}
\end{figure}
Since resource blocks are orthogonal to each other, we can
consider each resource block separately. Let the indices of the
users scheduled to resource block $b$ be $1_b$ and $2_b$ with
$1_b,2_b \in \{1,2,\cdots,K\}$, $b=1,2,\cdots,B$, and let the
channel gains from the BS to users $1_b$ and $2_b$ be
$h_{01_b}^{(b)}$ and $h_{02_b}^{(b)}$, respectively. (Here,
$h_{0k}^{(b)}$ is the channel gain from the BS to user $k$,
$k=1,2,\cdots,K$ at resource block $b$.) We assume that the
indices $1_b$ and $2_b$ are ordered such that $|h_{01_b}^{(b)}|
\ge |h_{02_b}^{(b)}|$. We assume that the BS and users $1_b$ and
$2_b$ form a RBC with user $1_b$ acting as a relaying receiver.
Then, the received signals $Y_1^{(b)}$ and $Y_2^{(b)}$ at users
$1_b$ and $2_b$ in resource block $b$ are respectively given by
\begin{align}
Y_1^{(b)} &= h_{01_b}^{(b)}X_0^{(b)} + Z_1^{(b)},   \label{eq:yibdatamodel}\\
Y_2^{(b)} &= h_{02_b}^{(b)}X_0^{(b)} + h_{1_b2_b}^{(b)}X_1^{(b)} +
Z_2^{(b)} \label{eq:yjbdatamodel}
\end{align}
where $h_{1_b2_b}^{(b)}$ represents the channel gain of the link
from user $1_b$ to user $2_b$, $X_0^{(b)}$ is the transmit signal
at the BS in resource block $b$  with power constraint
$\Ebb\{|X_0^{(b)}|^2\} \le P_{0}^{(b)}$, $X_1^{(b)}$ is the
transmit signal at the relaying user $1_b$ in resource block $b$
with power constraint $\Ebb\{|X_1^{(b)}|^2\} \le P_1^{(b)}$,  and
$Z_{1}^{(b)}\sim\Cc\Nc(0,N_{1}^{(b)})$ and
$Z_{2}^{(b)}\sim\Cc\Nc(0,N_{2}^{(b)})$ are the zero-mean additive
circularly-symmetric complex Gaussian noise at users $1_b$ and
$2_b$ in resource block $b$, respectively. Let the rates of users
$1_b$ and $2_b$ for resource block $b$ be $R_{1}^{(b)}$ and
$R_{2}^{(b)}$, respectively. Then, the overall system sum rate
$R_{sum}$ is given by
\begin{equation}
R_{sum}= \sum_{i=1}^B (R_{1}^{(b)} + R_{2}^{(b)}).
\end{equation}
The system sum rate $R_{sum}$ is a function of the component
channel rates $(R_1^{(b)},R_2^{(b)})$ for resource block $b$ and the
scheduling and grouping method.

\section{Component Channel Analysis: The Relaying Broadcast Channel}
\label{sec:componentchannel}

In this section, we analyze the achievable rate region of a
component RBC composed of the BS and users $1_b$ and $2_b$ for
each resource block $b$, which is the backbone for the later stage of this
paper.  As mentioned already, we consider DF and CF relaying for
user $1_b$ since the relative performance of DF and CF depends on the channel situation but the rate of AF is always worse than the better of DF and CF
 \cite{Kramer&Gastpar&Gupta:05IT}. Note that a RBC is different
from a relay channel since the transmitter sends two information
messages: one for the relaying receiver and the other for the
other receiver. We shall call RBC with DF and CF RBC-DF and
RBC-CF, respectively.  In the following subsections, we
investigate the achievable rate regions of RBC-DF and RBC-CF in the
discrete memoryless channel case first. Based on the result in the
discrete memoryless channel case, we obtain the achievable rate
regions of RBC-DF and RBC-CF in the Gaussian channel case next.

\begin{figure*}[h]
\centerline{
\SetLabels
\L(0.21*-0.1) (a) \\
\L(0.75*-0.1) (b) \\
\endSetLabels
\leavevmode
\strut\AffixLabels{ \psfrag{a1}[t][t]{\small Transmitter}
\psfrag{a1}[t][t]{\small Transmitter} %
\psfrag{a2}[t][t]{\small Receiver 1} %
\psfrag{a3}[t][t]{\small } %
\psfrag{a4}[t][t]{\small Receiver 2} %
\psfrag{a5}[t][t]{\small } %
\psfrag{b1}[t][t]{\small Transmitter} %
 \psfrag{b2}[t][t]{\small Relaying} %
 \psfrag{b3}[t][t]{\small Receiver} %
\psfrag{b4}[t][t]{\small Second} %
\psfrag{b5}[t][t]{\small Receiver} %
\psfrag{0}[c]{\small 0} %
\psfrag{1}[c]{\small 1} %
\psfrag{2}[c]{\small 2} %
\scalefig{0.4}\epsfbox{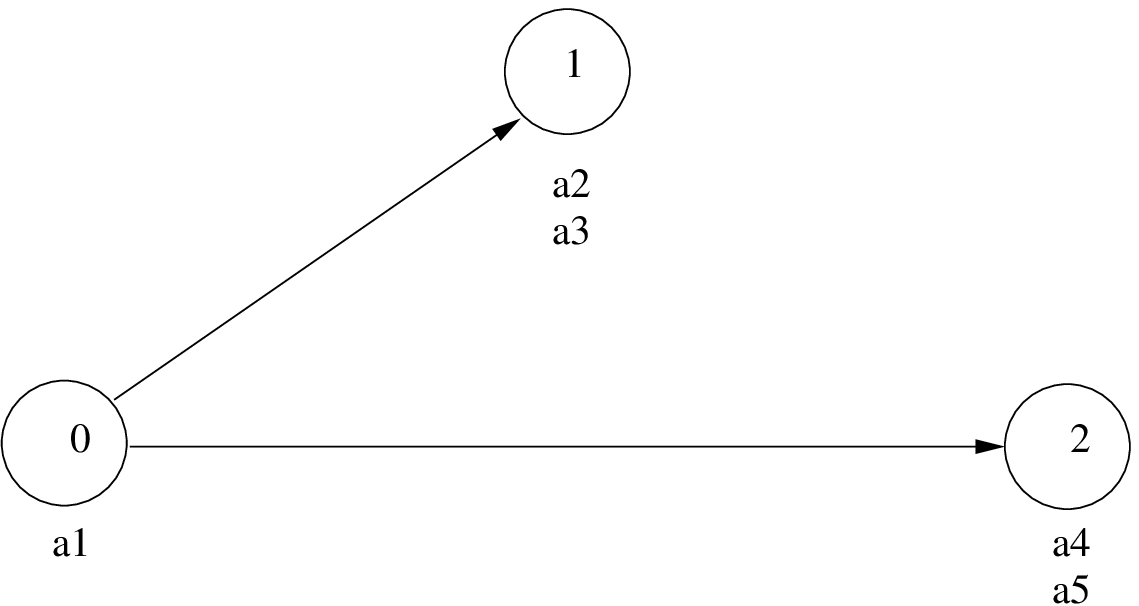}
\hspace{1cm} \scalefig{0.4}\epsfbox{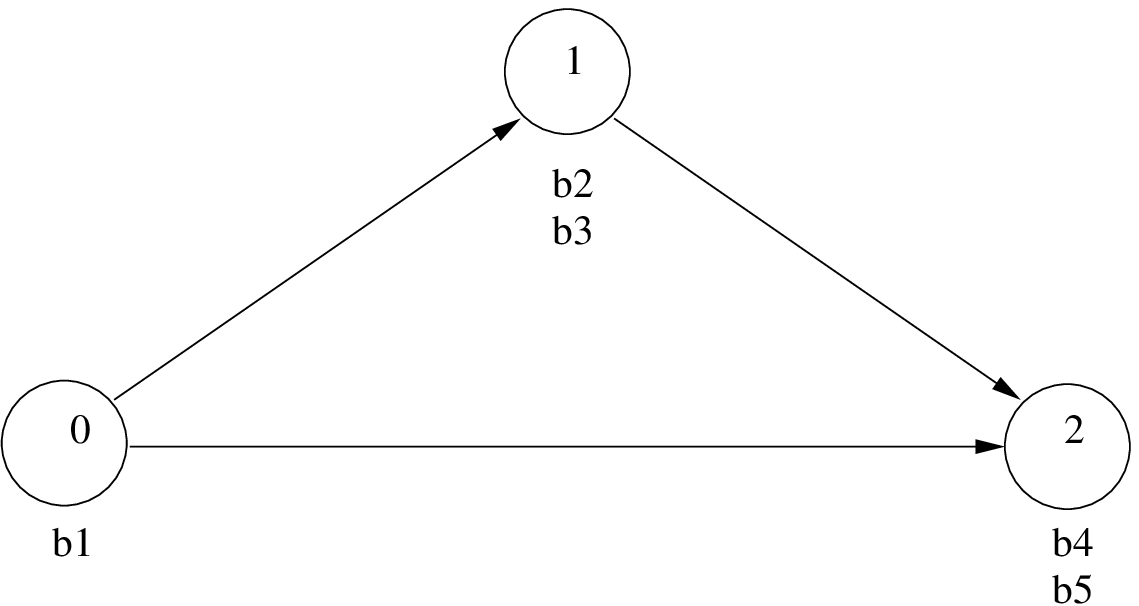} } }
\vspace{1.5em} \caption{(a) a broadcast channel (BC) and (b) a
relaying broadcast channel (RBC)} \label{fig:RBCchannel}
\end{figure*}

\subsection{The  Discrete Memoryless Case}
\label{sec:GeneralCase}

A general  RBC is a  3-node discrete memoryless network composed
of node 0 (the transmitter), node 1 (called the relaying
receiver), and node 2 (called the second receiver), as depicted in
Fig. \ref{fig:RBCchannel}(b), defined by
\begin{equation} \label{eq:RBCchanneldef}
(\Xc_0 \times
\Xc_1,p(y_1,y_2|x_0,x_1),\Yc_1\times\Yc_2),
\end{equation}
 where $\Xc_0$ and $\Xc_1$ are the input
alphabets of nodes 0 and 1, respectively;  $\Yc_1$ and $\Yc_2$ are the output alphabets of nodes 1 and 2, respectively;  and
$p(y_1,$ $y_2|x_0,x_1)$ is the channel transition probability mass
function.

From here on, we investigate the achievable rate region of the
considered RBC. First, we consider the case that node 1 does not
transmit  signal to node 2, i.e. $\Xc_1=\emptyset$. Then, the
channel reduces to a 2-user BC  and the capacity
region of a degraded BC is given by the following
theorem.

\vspace{0.5em}

\begin{theorem}\cite{Gamal&Kim:11NIT}\label{thm:dmbc}
The capacity region of the degraded discrete memoryless BC $(\Xc_0, p(y_1,y_2|x_0),$ $\Yc_1\times\Yc_2)$ is the set of rate pairs $(R_1,R_2)$ such that
\begin{align}
R_1 &< I(X_0;Y_1|U)  \label{eq:DMBC1}\\
R_2 &< I(U;Y_2) \label{eq:DMBC2}
\end{align}
for some pmf $p(u,x_0)$, where the cardinality of the auxiliary random variable $U$ satisfies $|\Uc|\le \min\{|\Xc_0|,|\Yc_1|,|\Yc_2|\}+1$. Here, $R_1$ and $R_2$ are the rates of nodes 1 and 2, respectively.
\end{theorem}

\vspace{0.5em}

It is known that the capacity region of a degraded BC can be achieved by superposition coding and SIC. This can
be seen in the rate formulae \eqref{eq:DMBC1} and
\eqref{eq:DMBC2}. Here, the auxiliary random variable $U$ is
associated with the message to node 2. In \eqref{eq:DMBC1}, $R_1$
is bounded by the mutual information between the transmitted
signal variable $X_0$ and node 1's received signal variable $Y_1$
conditioned on $U$. Conditioning can be viewed interference
cancellation and means that node 1 decodes the message associated
with node 2. On the other hand, $R_2$ is bounded simply by the
mutual information between its message variable $U$ and its
received signal variable $Y_2$.  The above capacity region result
is used in the simple NOMA
\cite{Saito&Kishiyama&Benjebbour&Nakamura&Li&Higuchi:13VTC}.

Now, consider the RBC scheme with DF at node 1. In this case,
contrary to the 2-user BC, we have
$\Xc_1\ne\emptyset$, which means that  node 1  not only decodes
the data for itself  but also actively helps  node 2. When the DF
relaying scheme is applied to the considered RBC, we have the
following rate region result given by
\cite{Liang&Veeravalli:04ISIT}:

\vspace{0.5em}

\begin{theorem}\cite{Liang&Veeravalli:04ISIT}\label{thm:rbcdf}
The rate pair $(R_1,R_2)$ is achievable for the RBC \eqref{eq:RBCchanneldef} if
\begin{align}
R_1 &< I(X_0;Y_1|U,X_1)   \label{eq:LVeq1}\\
R_2 &< \min\{I(U;Y_1|X_1), I(U,X_1;Y_2)\} \label{eq:LVeq2}
\end{align}
for some joint distribution $p(x_1)p(u|x_1)p(x_0|u)$, where $U$ is an auxiliary random variable associated with the message for node 2.
\end{theorem}

\vspace{0.5em}

The achievable rate region in Theorem \ref{thm:rbcdf} can also be
obtained by using superposition coding at node 0 and SIC at node 1
similarly to the result in Theorem \ref{thm:dmbc} and in addition
by node 1's transmitting  the decoded (at node 1) data (intended
for node 2) to node 2.    In Theorem \ref{thm:rbcdf},   $U$ is an
auxiliary random variable associated with the message for node 2
and $X_0$ is the input random variable at node 0 associated with
the messages for both nodes 1 and 2.  In  (\ref{eq:LVeq1}),
conditioning on $U$ means that node 1 decodes the message
associated with node 2, and then the mutual information between
$X_0$ and $Y_1$ conditioned on $U$  is related to the rate of the
message for node 1. (Here, node 1 knows its own transmit variable
$X_1$ and thus node 1 can cancel the self-interference. This is
seen as conditioning on $X_1$ in (\ref{eq:LVeq1}).)  Regarding
\eqref{eq:LVeq2}, the first term in the right-hand side (RHS) in
\eqref{eq:LVeq2} means that the message intended for node 2 should
be decoded successfully at node 1 for DF operation and the second
term in the RHS in \eqref{eq:LVeq2} is related to the rate at
which  node 2 decodes its message ($U$) based on its received
signal $Y_2$ with the help $(X_1)$ from node 1. Taking minimum in
\eqref{eq:LVeq2} means both events should happen in this scheme.
 Note that if we
remove $X_1$, \eqref{eq:LVeq1} reduces to  \eqref{eq:DMBC1}, and
the second term in the RHS of \eqref{eq:LVeq2} reduces to
\eqref{eq:DMBC2}. Here, the first term $I(U;Y_1)$ in the RHS of
\eqref{eq:LVeq2} without $X_1$ is always larger than or equal to
the second term $I(U;Y_2)$ in the RHS of  \eqref{eq:LVeq2} without
$X_1$, i.e., $I(U;Y_1) \ge I(U;Y_2)$ due to the assumption of
degradedness $U \rightarrow Y_1 \rightarrow Y_2$. Therefore, the
achievable rate region in Theorem \ref{thm:rbcdf} always subsumes
the capacity region in Theorem \ref{thm:dmbc}. In other words,
NOMA with the proposed RBC-DF is always better than the simple
NOMA adopting the degraded BC as its component channel in
\cite{Saito&Kishiyama&Benjebbour&Nakamura&Li&Higuchi:13VTC}.

When node 1 can decode the data intended for node 2, RBC-DF always
performs better than RBC-AF since the correct message for node 2
is regenerated at node 1 and forwarded to node 2, but in RBC-AF
node 1 only forwards a noise-corrupted version of the message for
node 2 directly to node 2. Note the rate $R_2$ in \eqref{eq:LVeq2}
for RBC-DF is limited by the term $I(U;Y_1|X_1)$ resulting from
the requirement that node 2's message should be decoded
successfully at node 1 for DF operation. One way to circumvent
this full decoding requirement is the CF scheme in which node 2's
information is compressed at node 1 and forwarded to node 2
\cite{Cover&Gamal:79IT}. It is known that CF outperforms DF under
certain situations \cite{Kramer&Gastpar&Gupta:05IT}.  When full
decoding of node 2's data at node 1 is not possible or results in
a low rate, we can resort to RBC-CF. Furthermore, RBC-CF performs
better than RBC-AF since AF is worse than CF
\cite{Kramer&Gastpar&Gupta:05IT}. Thus, we consider RBC-CF
adopting CF at node 1 as our next choice for the component
channel.
 An achievable rate region of a RBC-CF with common information intended for both nodes 1 and 2 was derived in \cite{Bross:09IT}.
 However, the derivation and the encoding scheme are  complicated and do not provide much insight.
Hence, we here simplify the problem by eliminating  common
information and derive a simple achievable rate region of a RBC-CF
based on the recent encoding and compression technique of noisy
network coding (NNC) presented in \cite{Lim&Kim&Gamal&Chung:11IT}.
The NNC in the 3-node setup is a simplified CF scheme compared to
the original CF scheme proposed in \cite{Cover&Gamal:79IT}.
Although we use the coding technique in  NNC, there is a
fundamental difference between   NNC and RBC. In the NNC setup,
the transmitter has only one message which may be intended for
multiple receivers. In RBC, however, the transmitter  has two
messages: one for the relaying receiver and the other for the
second receiver. By extending the NNC scheme to RBC, we obtain the
following result regarding the achievable rate region of a RBC.

\vspace{0.5em}

\begin{theorem}\label{prop:rbccf}
The rate pair $(R_1,R_2)$ is achievable for the RBC \eqref{eq:RBCchanneldef} if
\begin{align}
R_1 &< I(U;Y_1) \label{eq:RBC_CF_Achievable2}\\
R_2 &< \min\{I(V;\hat{Y}_1,Y_2|X_1),I(V,X_1;Y_2)-I(\hat{Y}_1;Y_1|V,X_1,Y_2)\}\label{eq:RBC_CF_Achievable}
\end{align}
for some joint distribution
\begin{equation} \label{eq:RBCCF_genDist}
p(x_1)p(u)p(v)p(x_0|u,v)p(\hat{y}_1|y_1,x_1).
\end{equation}
\end{theorem}
\begin{proof}
See Appendix \ref{subsec:appendA}.
\end{proof}

\vspace{0.5em}

Here, $U$ and $V$ are the input message variables to node 1 (the
relaying receiver) and node 2 (the second receiver), respectively,
and the overall transmit variable $X_0$ of node 0 is generated
based on $(U,V)$,  as seen in the term $p(x_0|u,v)$ in the
generating input distribution in \eqref{eq:RBCCF_genDist}. Thus,
the rate $R_1$ is simply the mutual information between the
message variable $U$ for node 1 and the received signal variable
$Y_1$ at node 1.  The rate $R_2$ is the rate of NNC with $V$ as
the transmit variable at node 0, where the cut-set bound is used
\cite{Lim&Kim&Gamal&Chung:11IT,Gamal&Kim:11NIT}.  The first term
$I(V;\hat{Y}_1,Y_2|X_1)$ in the RHS of
\eqref{eq:RBC_CF_Achievable} is the mutual information between
node 0 and nodes $\{1,2\}$ with self interference cancellation at
the cut group, nodes $\{1,2\}$.  The term
$I(V,X_1;Y_2)$\footnote{This term corresponds to the second term
$I(U,X_1;Y_2)$ in the RHS of \eqref{eq:LVeq2} in the RBC-DF
scheme.} in the second term in the RHS of
\eqref{eq:RBC_CF_Achievable} is the decoding rate of node 2 with
the help ($X_1$) from node 1 and the term
$I(\hat{Y}_1;Y_1|V,X_1,Y_2)$ in the second term in the RHS of
\eqref{eq:RBC_CF_Achievable} represents the loss related to
compression compared to full decoding. For the details of the
encoding and decoding scheme for the rate-tuple in Theorem
\ref{prop:rbccf}, see Appendix \ref{subsec:appendA}.

\subsection{The Gaussian Case}
\label{sec:GaussianCase}

In this section, we consider the Gaussian channel case and compare
the performance of the  three component channel formulation
schemes: GBC (simple NOMA), RBC-DF, and RBC-CF/NNC. In the
Gaussian channel case, the received signals at the relaying
receiver and the second receiver  are given by
\eqref{eq:yibdatamodel} and
 \eqref{eq:yjbdatamodel}, respectively, which are rewritten here as
\begin{align}
Y_1 &= h_{01}X_0 + Z_1,  \label{eq:GaussianChannel1}\\
Y_2 &= h_{02}X_0 + h_{12}X_1 + Z_2, \label{eq:GaussianChannel2}
\end{align}
where the resource block superscript $(b)$ is omitted. Here, $Y_1$
and $Y_2$ are the received signals at the relaying receiver and
the second receiver, respectively; $X_0$ and $X_1$ are the
transmit signals from the transmitter and the relaying receiver,
respectively; $h_{ij}$ denotes the channel from node $i$ to node
$j$; and $Z_i\sim\Cc\Nc(0,N_i)$ is the zero-mean additive Gaussian
noise at node $i$.

To compute the rate-tuples in Theorems \ref{thm:dmbc},
\ref{thm:rbcdf}, and \ref{prop:rbccf}, we need to specify the
associated  input distributions since the channel
$p(y_1,y_2|x_0,x_1)$ is given.  We set $X_0\sim\Cc\Nc(0,P_0)$ and
$X_1\sim\Cc\Nc(0,P_1)$, and set the transmitted signal at node 0
(i.e., the transmitter) as the superimposed signal given by
\begin{align}  \label{eq:GaussianSuperpose}
X_0= U + V,
\end{align}
where $U$ is the signal for node $1$ and $V$ is the signal for node
$2$:
\begin{equation} \label{eq:UVdistrGauss}
U \sim \Cc\Nc(0,\alpha P_0), ~~~   V \sim \Cc\Nc(0,
\bar{\alpha} P_0), ~~~\bar{\alpha}=1-\alpha,~~~ 0\le \alpha \le
1.
\end{equation}

It is known that a two-user SISO GBC is a degraded BC since either
$\frac{|h_{01}|^2}{N_1} \ge \frac{|h_{02}|^2}{N_2}$ or
$\frac{|h_{01}|^2}{N_1} < \frac{|h_{02}|^2}{N_2}$. With the
considered ordering in Section \ref{sec:systemmodel}, we have
$\frac{|h_{01}|^2}{N_1} \ge \frac{|h_{02}|^2}{N_2}$. Then, the
following NOMA condition is automatically satisfied:
\begin{align} \label{eq:noma_condition}
\frac{|h_{01}|^2\bar{\alpha}P_0}{|h_{01}|^2\alpha P_0 + N_1}\ge\frac{|h_{02}|^2\bar{\alpha}P_0}{|h_{02}|^2\alpha P_0 + N_2}.
\end{align}
The capacity region of GBC (simple NOMA) is given by
\begin{align}
R_1 &\le \log\left(1+\frac{|h_{01}|^2\alpha P_0}{N_1}\right),\label{eq:noma_1}\\
R_2 &\le \log\left(1+\frac{|h_{02}|^2\bar{\alpha}P_0}{|h_{02}|^2\alpha P_0 + N_2}\right).\label{eq:noma_2}
\end{align}
Next, consider the RBC-DF scheme. From Theorem \ref{thm:rbcdf}, the achievable rate region is given by
\begin{align}
R_1 &\le \log\left(1+\frac{|h_{01}|^2\alpha P_0}{N_1}\right),\label{eq:rbcdf_1}\\
R_2 &\le \min\left\{\log\left(1+\frac{|h_{02}|^2\bar{\alpha}P_0 + |h_{12}|^2P_1}{|h_{02}|^2\alpha P_0 + N_2}\right), ~\log\left(1+\frac{|h_{01}|^2\bar{\alpha}P_0}{|h_{01}|^2\alpha P_0 + N_1}\right)\right\}.\label{eq:rbcdf_2}
\end{align}
From the fact that the rates \eqref{eq:noma_1} and
\eqref{eq:rbcdf_1} for $R_1$  are the same and \eqref{eq:noma_2}
for $R_2$ is always smaller than or equal to \eqref{eq:rbcdf_2}
for $R_2$ by the condition \eqref{eq:noma_condition}, we can
easily see that the achievable rate region of RBC-DF subsumes the
capacity region of simple NOMA based on GBC. The improvement of
rate $R_2$ is large when $|h_{12}|^2P_1$ is large and the gap
between $\frac{|h_{01}|^2}{N_1}$ and $\frac{|h_{02}|^2}{N_2}$ is
large.

Now, consider RBF-CF/NNC in the Gaussian case.  Here we use
Theorem \ref{prop:rbccf} to derive an achievable rate region in
the Gaussian case. Note that the input distribution in this case
is given by $p(x_1)p(u)p(v)p(x_0|u,v)p(\hat{y}_1|y_1,x_1)$ in
\eqref{eq:RBCCF_genDist}. Thus, to apply Theorem \ref{prop:rbccf}
to the Gaussian channel case,  we further set the remaining part
$p(\hat{y}_1|y_1,x_1)$ of the input distribution as
\begin{align}   \label{eq:GaussianCFhatY}
\hat{Y}_1 &= Y_1 - h_{01}U + \hat{Z} = h_{01}V + Z_1 + \hat{Z},
\end{align}
where $\hat{Z}\sim\Cc\Nc(0,\hat{N})$.
With some calculation, we get the following achievable rate region of RBC-CF/NNC:
\begin{align}
R_1 \le &\log\left(1+\frac{|h_{01}|^2\alpha P_0}{|h_{01}|^2\bar{\alpha}P_0 + N_1}\right),\label{eq:rbccf_1}\\
R_2 \le &\min\left\{\log\left(1+\frac{(|h_{02}|^2\alpha P_0+N_2)|h_{01}|^2\bar{\alpha}P_0 + (N_1 +\hat{N})|h_{02}|^2\bar{\alpha}P_0}{(N_1+\hat{N})(|h_{02}|^2\alpha P_0+N_2)}\right),\right.\nonumber\\
        &\log\left(1+\frac{|h_{02}|^2\bar{\alpha}P_0+|h_{12}|^2P_1}{|h_{02}|^2{\alpha}P_0+N_2}\right) \nonumber\\
        &\left.-\log\left(1+\frac{N_1^2N_2+N_1^2|h_{02}|^2\alpha P_0}{\hat{N}N_1N_2+\hat{N}N_2|h_{01}|^2\alpha P_0+\hat{N}N_1|h_{02}|^2\alpha P_0+N_1N_2|h_{01}|^2\alpha P_0}\right)\right\}  \label{eq:rbccf_rateR2}
\end{align}
(The detail of the calculation is in Appendix
\ref{subsec:appendB}.) The rate $R_2$ of the second receiver in
\eqref{eq:rbccf_rateR2} can be larger than that of RBC-DF in
\eqref{eq:rbcdf_2} depending on the situation. However, note  that
the rate $R_1$ of the relaying receiver in \eqref{eq:rbccf_1} is
smaller than that of GBC and RBC-DF. This is because at the
relaying receiver the message for the second receiver is not fully
decoded and thus the interference from the second receiver's
signal at the relaying receiver cannot be cancelled by SIC. To
resolve this problem, we apply DPC
\cite{Costa:83IT} at the transmitter together with the encoding
scheme presented in Theorem \ref{prop:rbccf}
 to remove the
interference from the second receiver's signal at the relaying
receiver since the transmitter knows both messages
\cite{Gelfand&Pinsker:80PCIT}. In this case, the transmitter
generates the message codeword for the second receiver first and
then based on this message codeword it generates the message
codeword for the relaying receiver based on DPC. Then, the
transmitter superimposes the two message codewords and transmits
the superimposed signal. The processing at the relaying receiver
and the second receiver is the same as RBC-CF/NNC. In the decoding
process of the relaying receiver for its own message, the
interference from the second receiver's signal is automatically
removed due to DPC applied at the transmitter side. The achievable
rate region of RBC-CF/NNC employing DPC is given by
\begin{align}
R_1 \le &\log\left(1+\frac{|h_{01}|^2\alpha P_0}{N_1}\right) \label{eq:rbccf_rateR1dpc}\\
R_2 \le &\min\left\{\log\left(1+\frac{(|h_{02}|^2\alpha P_0+N_2)|h_{01}|^2\bar{\alpha}P_0 + (N_1 +\hat{N})|h_{02}|^2\bar{\alpha}P_0}{(N_1+\hat{N})(|h_{02}|^2\alpha P_0+N_2)}\right),\right.\nonumber\\
        &\log\left(1+\frac{|h_{02}|^2\bar{\alpha}P_0+|h_{12}|^2P_1}{|h_{02}|^2{\alpha}P_0+N_2}\right) \nonumber\\
        &\left.-\log\left(1+\frac{N_1^2N_2+N_1^2|h_{02}|^2\alpha P_0}{\hat{N}N_1N_2+\hat{N}N_2|h_{01}|^2\alpha P_0+\hat{N}N_1|h_{02}|^2\alpha P_0+N_1N_2|h_{01}|^2\alpha
        P_0}\right)\right\}.\label{eq:rbccf_rateR2dpc}
\end{align}
Note that in this scheme $R_2$ is the same as
\eqref{eq:rbccf_rateR2} of RBC-CF/NNC but $R_1$ is improved to be
the same as \eqref{eq:rbcdf_1} of GBC and RBC-DF. The value of
$\hat{N}$ can be optimized to yield maximum $R_2$ in
\eqref{eq:rbccf_rateR2} and \eqref{eq:rbccf_rateR2dpc} by solving
a quadratic equation.  The proposed encoding scheme based on both
superposition/DPC and CF/NNC for NOMA is described in Fig.
\ref{fig:codeMarkov_Chain} in Appendix \ref{subsec:appendA}.

\section{The Considered User Scheduling and Pairing}
\label{sec:pairing}

In Section \ref{sec:componentchannel}, we have investigated the
achievable regions for several component channel types. In this
section, we introduce two user pairing methods  to compare the
performance of the overall system adopting one of the considered
component channel types: GBC (simple NOMA), RBC-DF or RBC-CF as
the component channel. Since the performance of the overall system
depends on user pairing, we consider two disparate user pairing
methods: near-far pairing and nearest neighbor pairing. The two
pairing methods are opposite to each other and are useful to
compare NOMA employing a different component channel type in
different system setting.

\subsection{Near-Far Pairing}
\label{subsec:grouping}

The first considered user scheduling and pairing is similar to
that in
\cite{Benjebbour&Li&Saito&Kishiyama&Harada&Nakamura:13Globecom}
except that we consider a sequential approach. In the first
method, we aim at pairing two users: one with good channel and the
other with bad channel. We assume that the power for the relaying
receiver and the power for the second receiver for each resource
block are fixed, i.e., the parameter $\alpha$ in
\eqref{eq:UVdistrGauss}  is given, and the BS knows the location
of each user in the cell and the gain of the channel from the BS
itself to each user in the cell, i.e., $h_{0k}^{(b)}$ for
$k=1,2,\cdots,K$ and $b=1,2,\cdots,B$. First, the users in the
cell are divided into two groups for each resource block $b$:
group $G_1^{(b)}$ with good channel with $K/2$ users and group
$G_2^{(b)}$  with bad channel with $K/2$ users by ordering
$|h_{0k}^{(b)}|$ for each resource block $b$. Then, for resource
block $b=1$, we pick one user (which becomes the relaying
receiver) from $G_1^{(1)}$  based on the proportionally fair (PF)
scheduling \cite{Viswanath&Tse&Laroia:02IT} and the instantaneous
achievable rate $R_1$ given in Section \ref{sec:GaussianCase} for
RBC-DF, RBC-CF/NNC, or RBC-CF/NNC/DPC. That is, the selected user
is given by
\begin{align}
\kappa_1^{(1)} &= \mathop{\arg\max}\limits_{i\in G_1^{(1)}}   \frac{R_{1(i)}^{(b)}[t]}{\bar{\Rc}(i)[t]},
\end{align}
where $R_{1(i)}^{(b)}[t]$ is the rate $R_1$ given in Section \ref{sec:GaussianCase} when user $i$ serves as the relaying receiver at time $t$ and resource block $b$, and $\bar{\Rc}(i)[t]$ is the average served rate for user $i$ up to time $t$. Note from Section \ref{sec:GaussianCase} that $R_{1(i)}^{(b)}[t]$ can be computed based only on $h_{0i}^{(b)}$.  After $\kappa_1^{(1)}$ is chosen, we select the second user $\kappa_2^{(1)}$ for resource block $b=1$ from $G_2^{(1)}$ based on $\kappa_1^{(1)}$  and again the PF principle, i.e.,
\begin{align}
\kappa_2^{(1)} &= \mathop{\arg\max}\limits_{i\in G_2^{(1)}} \frac{R_{2(i|\kappa_1^{(1)} )}^{(b)}[t]}{\bar{\Rc}(i)[t]},
\end{align}
where $R_{2(i|j)}^{(b)}[t]$ is the rate $R_2$ given in
Section \ref{sec:GaussianCase} when user $i$ is the second
receiver paired with the relaying receiver $j$ at time $t$ and
resource block $b$. Here, as seen in Section
\ref{sec:GaussianCase},  the computation of
$R_{2(i|\kappa_1^{(1)} )}^{(b)}[t]$ requires the knowledge
of the channel gain $|h_{\kappa_1^{(1)}i}^{(b)}|^2$ from user
$\kappa_1^{(1)}$ and user $i$.  In this step, we use an estimate
for the channel gain based on  \cite{Holma&Toskala:book}
\begin{equation}
\widehat{|h_{ij}^{(b)}|^2} = C_0 d^{-\gamma},
\end{equation}
where $C_0$ is a constant, $d$ is the distance between users $i$
and $j$, and $\gamma$ is the path loss exponent. (The assumption
of knowledge of user locations at the BS is required for this
step.) After $\kappa_1^{(1)}$ and $\kappa_2^{(1)}$ for resource
block $b=1$ are selected, we proceed to $b=2$. For resource block
$b=2$, we remove $\kappa_1^{(1)}$ and $\kappa_2^{(1)}$ from
$G_1^{(2)}$ and  $G_2^{(2)}$, and repeat the same procedure with
the remaining sets. After users are selected for all resource
blocks, we update the average served rate for the served users as
\begin{align}   \label{eq:PFforgetting}
\bar{\Rc}(i)[t+1] := (1-\tau)\bar{\Rc}(i)[t] + \tau R(i)[t], ~~~i=1,\cdots,K,
\end{align}
where $R(i)[t]$ is the served rate for user $i$ at time $t$, and
$\tau$ is the auto-regressive (AR) filter coefficient or
forgetting factor.

\subsection{Nearest-Neighbor Pairing}
\label{subsec:pairing}

The second scheduling and pairing is quite opposite to the first
method. In the second method, we aim at pairing two users who are
close to each other. The reason of considering the second pairing
method is to investigate the performance of general NOMA over a
wide range of user pairing methods. In the second method, we
select one user as the relaying receiver and its nearest neighbor
as the second receiver. Since the nearest neighbor for each user
is given, we can select the two users simultaneously based on the
PF metric. That is, for resource block $b=1$, set
$\Gc=\{1,2,\cdots,K\}$ and
\begin{align}
\kappa_1^{(1)} &= \mathop{\arg\max}\limits_{i \in \Gc}\left(\frac{R_{1(i)}^{(b)}[t]}{\bar{\Rc}(i)[t]}+\frac{R_{2(\Nc(i)|i)}^{(b)}[t]}{\bar{\Rc}(\Nc(i))[t]}\right)
\end{align}
where $\Nc(i)$ is the index of the nearest neighbor of user $i$.
When user selection for resource block $b=1$ is finished, we
remove $\kappa_1^{(1)}$ and $\kappa_2^{(1)}=\Nc(\kappa_1^{(1)})$
from set $\Gc$, and repeat the same procedure for resource blocks
$b=2,\cdots,B$.

\section{Numerical Results}
\label{sec:numericalresults}

In this section, we provide some numerical results to evaluate the
performance of the proposed NOMA with RBC. We first evaluate the
performance of each component channel presented in Section
\ref{sec:GaussianCase} and then evaluate the  sum rate of the
entire cell employing the considered user pairing and scheduling
presented in Section \ref{sec:pairing} and the considered RBC
component channel.

\subsection{The Component Channel Performance}
\label{subsec:nr_component_channel}

\begin{figure}[!hbtp]
\centerline{ \SetLabels
\L(0.25*-0.1) \small{(a)} \\
\L(0.76*-0.11) \small{(b)} \\
\endSetLabels
\leavevmode
\strut\AffixLabels{
\scalefig{0.45}\epsfbox{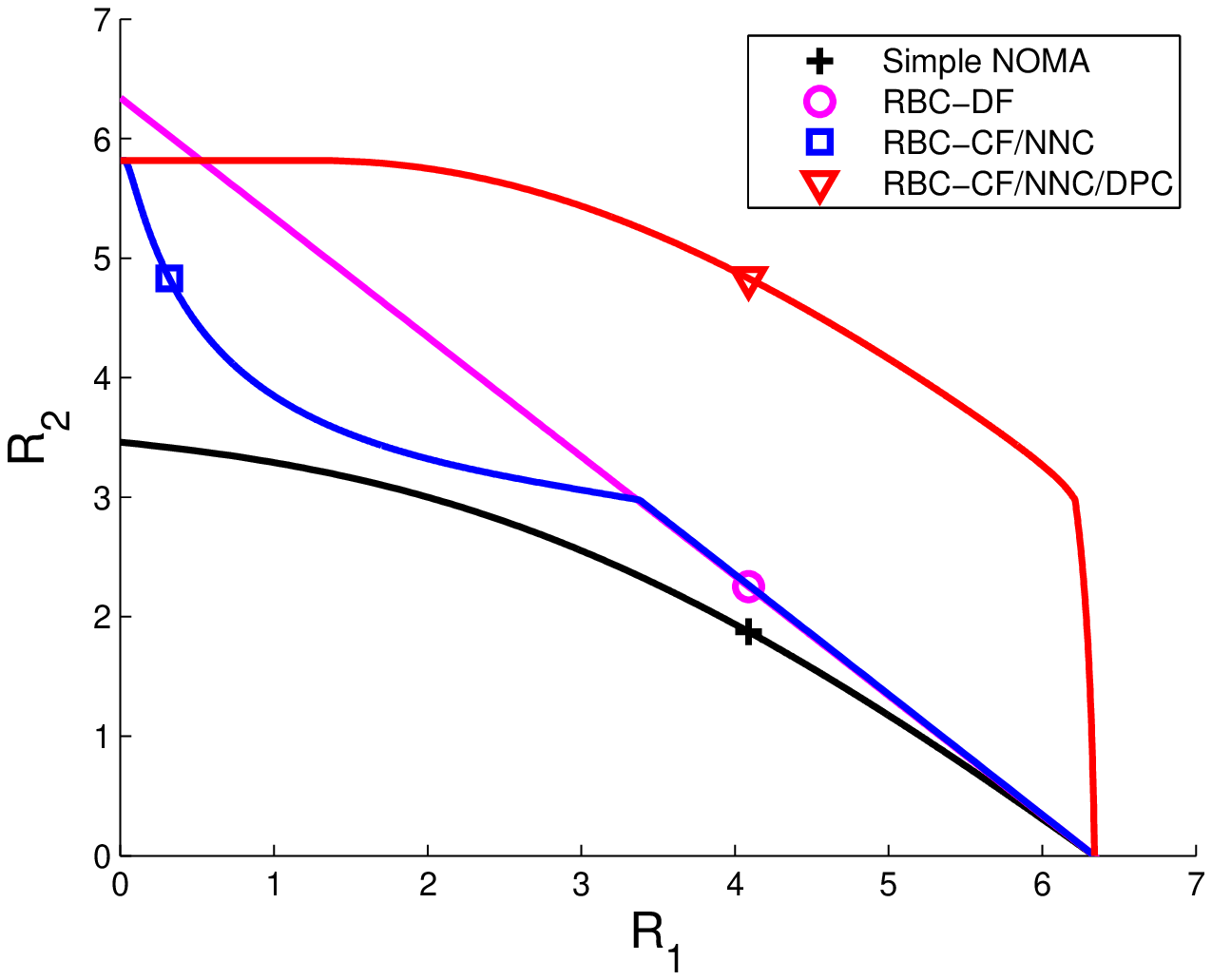}
\scalefig{0.45}\epsfbox{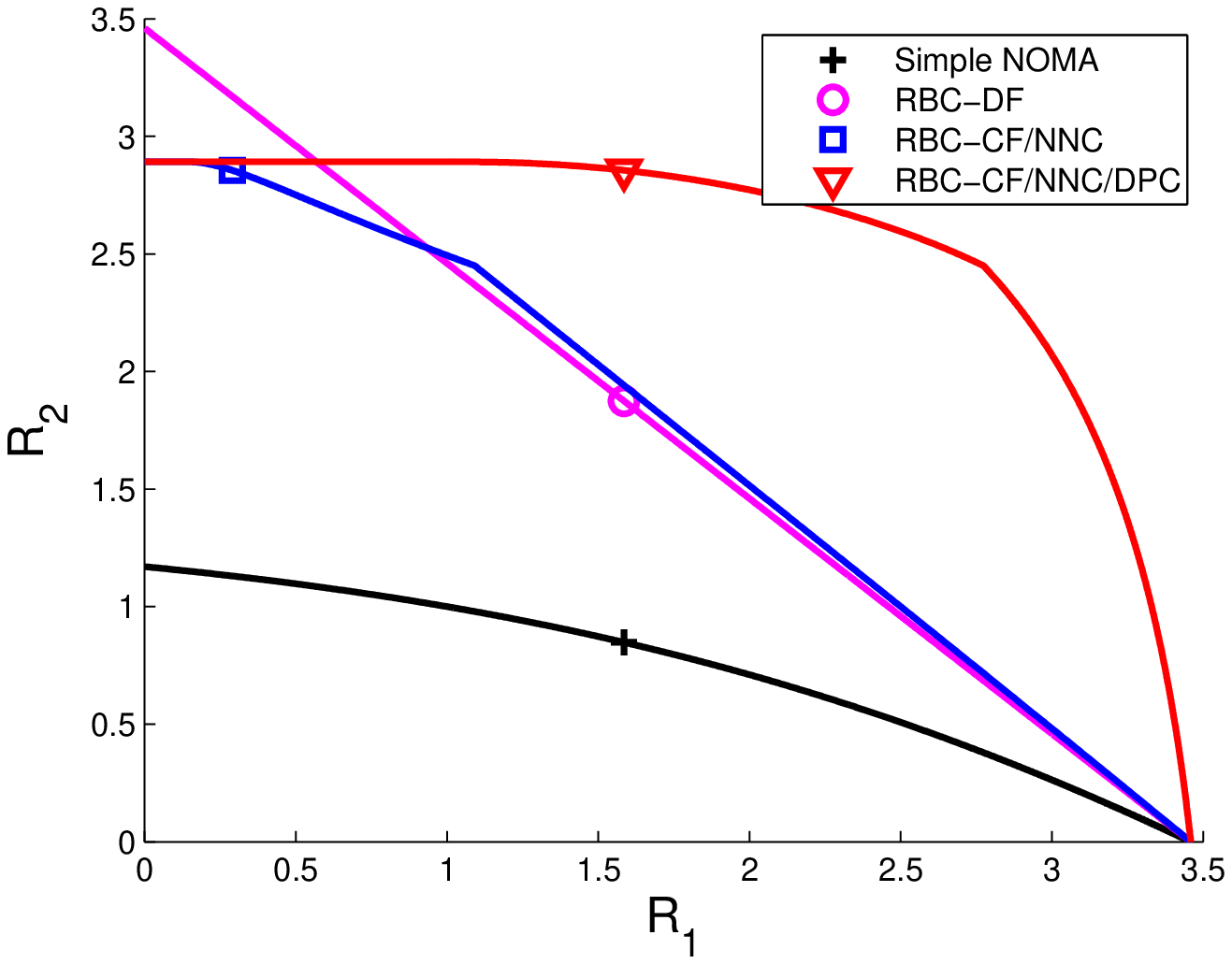} } }
\vspace{0.5cm} \centerline{ \SetLabels
\L(0.25*-0.1) \small{(c)} \\
\L(0.76*-0.1) \small{(d)} \\
\endSetLabels
\leavevmode
\strut\AffixLabels{
\scalefig{0.45}\epsfbox{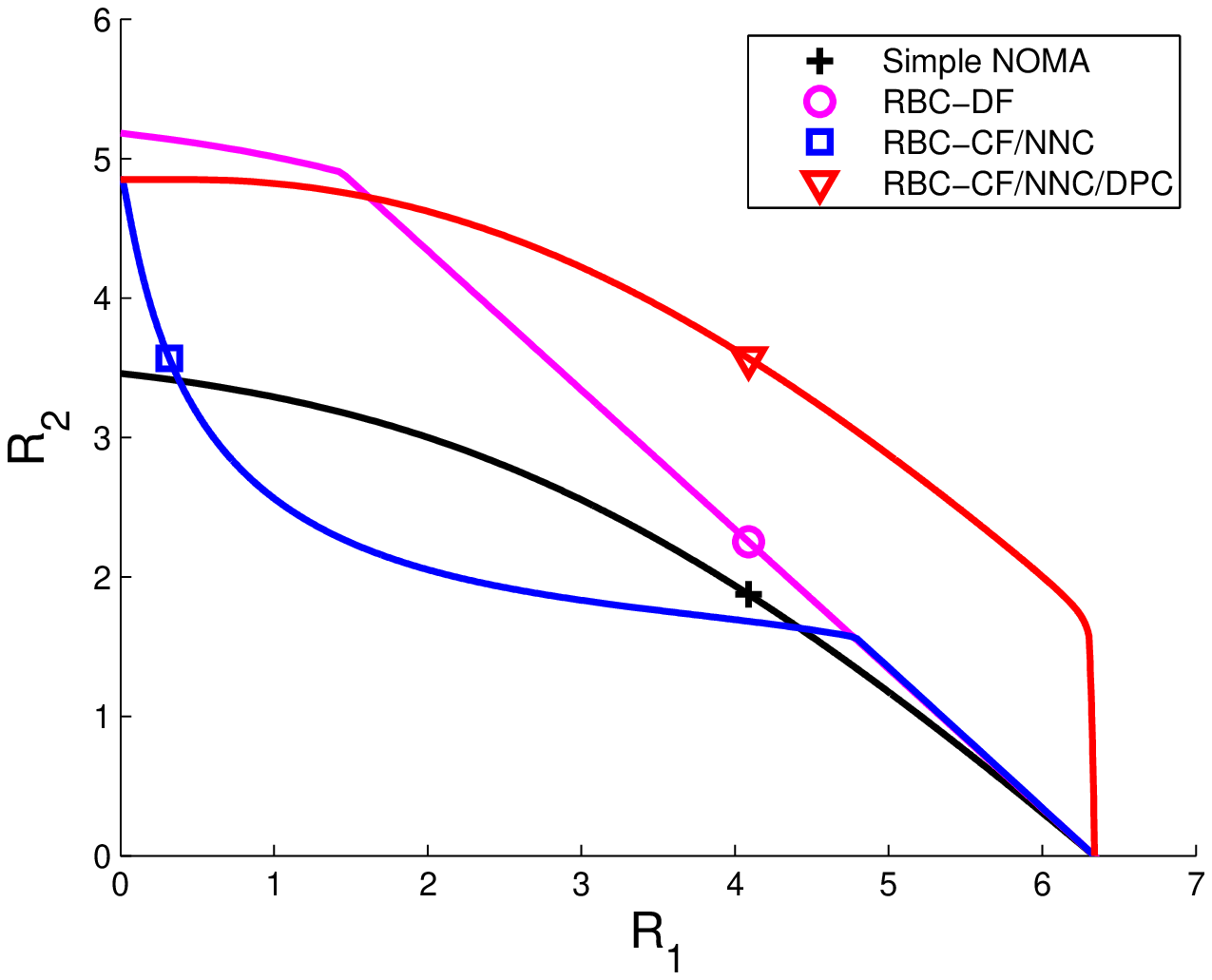}
\scalefig{0.45}\epsfbox{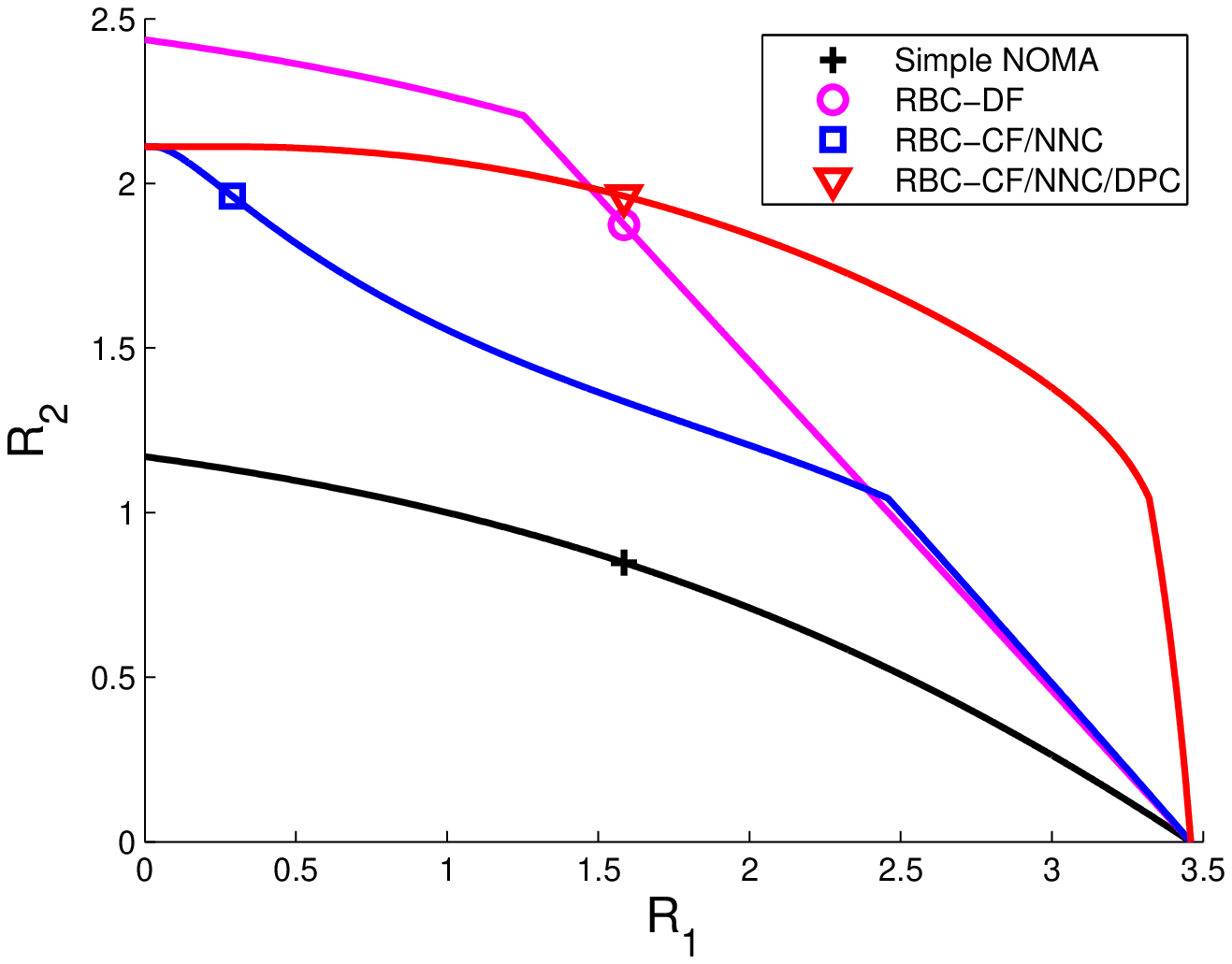} } }
\vspace{0.5cm} \centerline{ \SetLabels
\L(0.25*-0.1) \small{(e)} \\
\L(0.76*-0.1) \small{(f)} \\
\endSetLabels
\leavevmode
\strut\AffixLabels{
\scalefig{0.45}\epsfbox{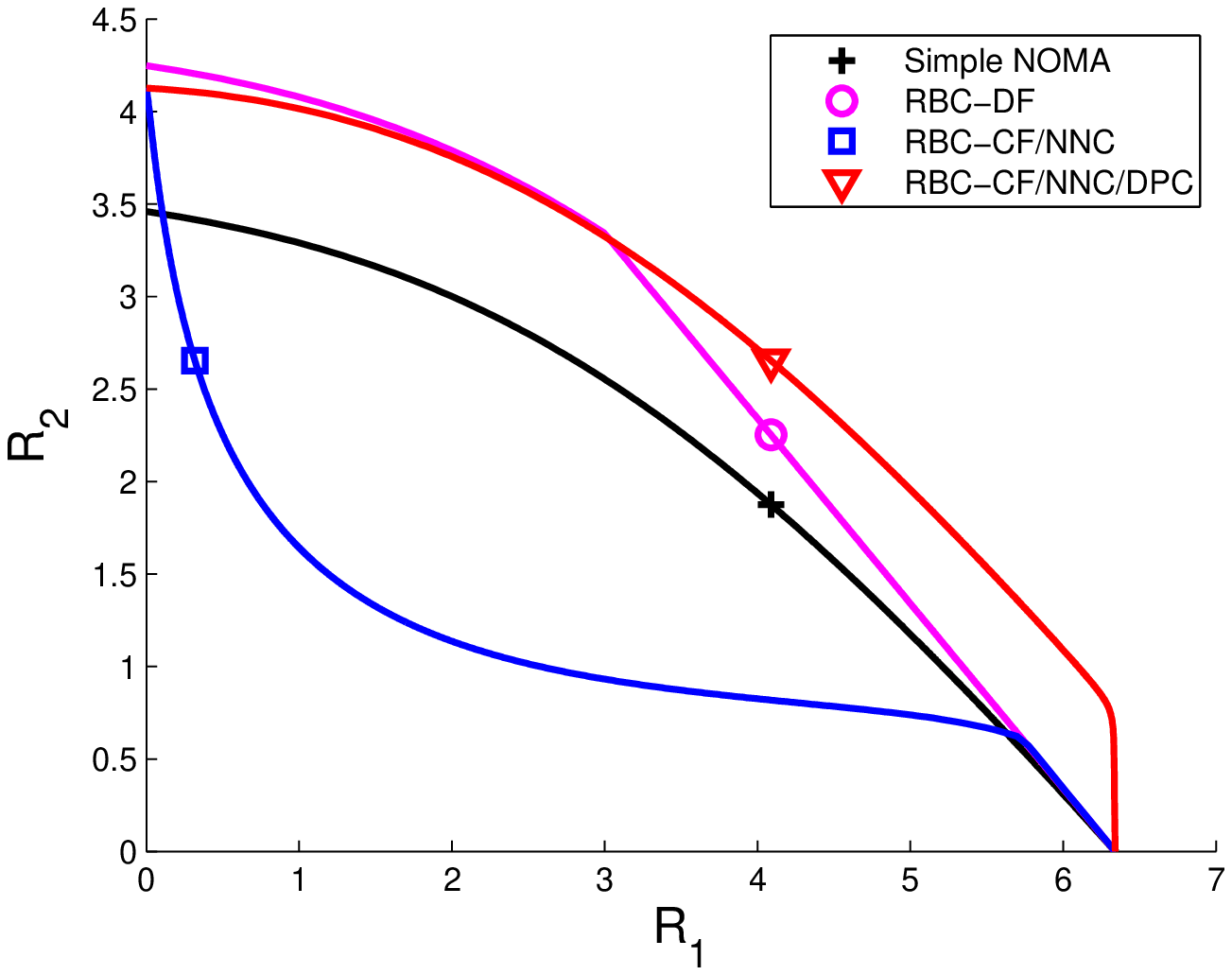}
\scalefig{0.45}\epsfbox{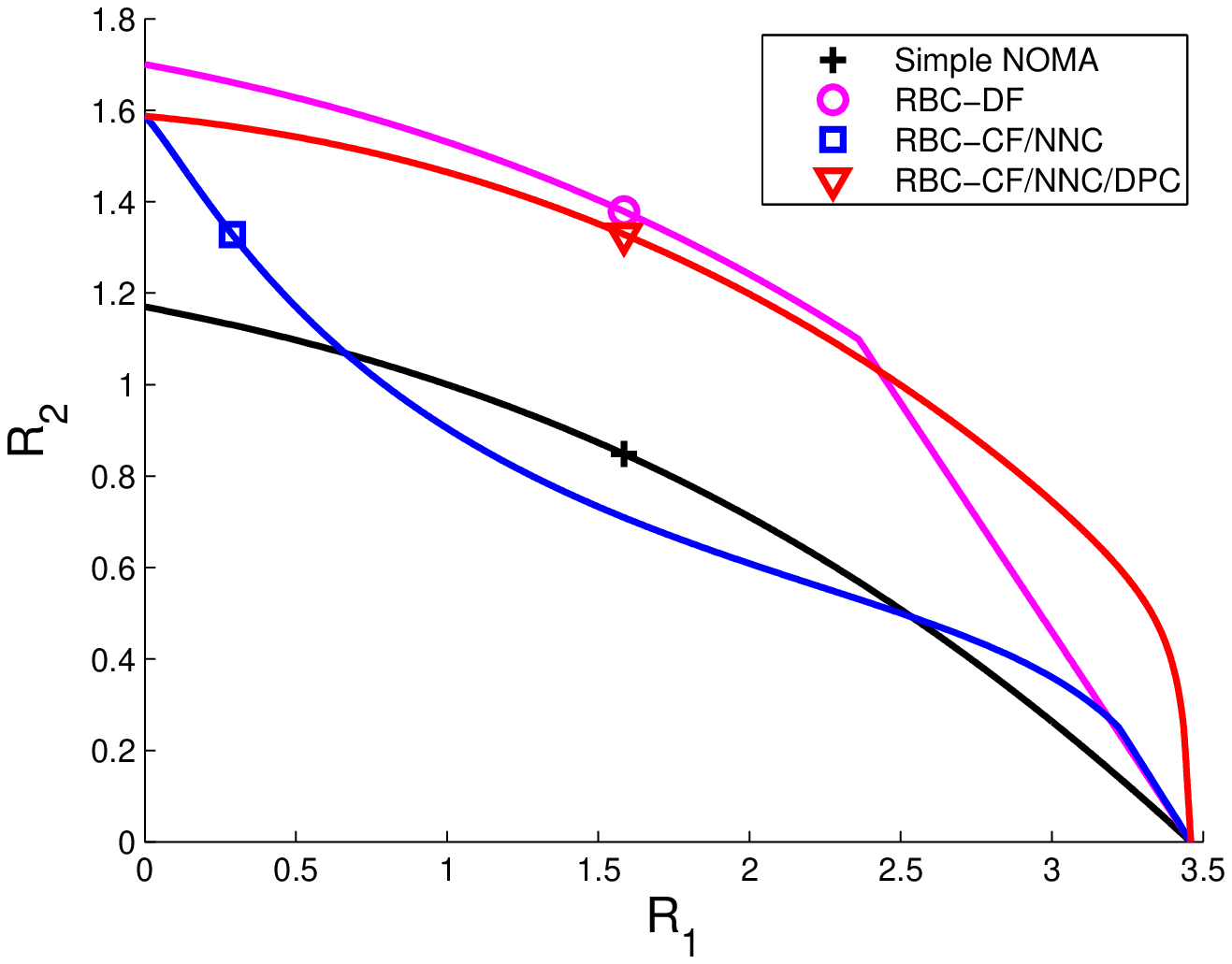} } }
\vspace{0.5cm} \caption{The achievable rate region -
$(|h_{01}|^2,|h_{02}|^2,|h_{12}|^2)=(8,8,1):(a)$, $(c)$ and $(e)$,
 $(|h_{01}|^2,|h_{02}|^2,|h_{12}|^2)=(1,1,1/8):(b)$, $(d)$ and
$(f)$.  $N=N_1=N_2=1$.    $(P_0^{(b)}/N,
P_1^{(b)}/N)=(10\mbox{dB},10\mbox{dB})$: (a) and (b),
$(P_0^{(b)}/N, P_1^{(b)}/N)=(10\mbox{dB},5\mbox{dB})$: (c) and
(d), $(P_0^{(b)}/N, P_1^{(b)}/N)=(10\mbox{dB},0\mbox{dB})$: (e)
and (f) } \label{fig:rateRegion}
\end{figure}

For the evaluation of the performance of component channels, we
considered  a linear configuration in which
 the location of the relaying receiver in the middle of the
line between the transmitter and the second receiver.  We set
$N_1=N_2=N=1$ and considered three pairs of  $(P_0^{(b)},
P_1^{(b)})$  for the transmit power $P_0^{(b)}$ at the BS and the
transmit power $P_1^{(b)}$ at the relaying receiver:
$(P_0^{(b)}/N, P_1^{(b)}/N)$ = (10 dB, 10 dB), (10 dB, 5dB), and
(10 dB, 0 dB).\footnote{In real-world cellular systems, the
maximum BS downlink average transmit power is  43 dBm (20W) and
the maximum average transmit power of a cellular phone is 24 dBm
(0.25W). However, the BS downlink transmit power is shared by 50
to 100 simultaneous users. Hence, the maximum per-user BS downlink
average power is around 23 dBm to 26 dBm. This is the basis for
the consider relative magnitude for
 $P_0^{(b)}$ and $P_1^{(b)}$.}
We assumed that the path loss exponent is $\gamma=3$. Based on
$\gamma=3$, we considered two channel gain setup:
(i)\footnote{With the channel gain setting (i), we have node 1's
received signal-to-noise ratio (SNR) $|h_{01}|^2\alpha
P_0^{(b)}/N_1=12$dB and node 2's received SNR
$|h_{02}|^2(1-\alpha) P_0^{(b)}/N_2=9$dB for $\alpha=0.2$, a
typical power distribution value in NOMA \cite{Benjebbour&Li&Saito&Kishiyama&Harada&Nakamura:13Globecom}. Node
2's SNR of 9dB is higher than the signal-to-interference ratio
(SIR) of 0.8/0.2=6dB. With the channel gain setting (ii), each
node's SNR is reduced by 9dB.} $|h_{01}|^2=|h_{12}|^2=8$ and
$|h_{02}|^2=1$ and (ii) $|h_{01}|^2=|h_{12}|^2=1$ and
$|h_{02}|^2=1/8$. Then, we swept the value of the parameter
$\alpha$ defined in \eqref{eq:UVdistrGauss} to determine the
achievable rate pair ($R_1,R_2$). The result is shown in Fig.
\ref{fig:rateRegion}. Fig.\ref{fig:rateRegion}(a), (c) and (e)
show the rate-tuples in [\eqref{eq:noma_1}, \eqref{eq:noma_2}: GBC
- simple NOMA], [\eqref{eq:rbcdf_1}, \eqref{eq:rbcdf_2}: RBC-DF],
[\eqref{eq:rbccf_1}, \eqref{eq:rbccf_rateR2}: RBC-CF/NNC], and
[\eqref{eq:rbccf_rateR1dpc}, \eqref{eq:rbccf_rateR2dpc}:
RBC-CF/NNC/DPC] for the channel gain setting (i) of around 10 dB
received SNR operation. It is seen that the proposed NOMA
 equipped with RBC component channels employing superposition/DPC and CF/NNC
significantly improves the performance over the simple NOMA based
on GBC/SIC. The marked points in Fig. \ref{fig:rateRegion} are the
rate-pair points of $\alpha=0.2$. It is seen that for
$\alpha=0.2$, $R_2$ of RBC-CF/NNC without DPC is higher than $R_2$
of RBC-DF but $R_1$ of RBC-CF/NNC without DPC is much lower  than
$R_2$ of RBC-DF, as expected.  It is also seen that in the channel
gain setting (i) of roughly 10 dB received SNR operation, the gain
of RBC-DF over GBC is not so large at $\alpha=0.2$.
Fig.\ref{fig:rateRegion}(b), (d) and (f) show the rate-tuples in
the channel gain setting (ii) of 0dB received SNR operation. It is
seen that the gain by the RBC-CF/NNC/DPC over the simple NOMA
(GBC) is drastic.

\subsection{The Overall System Performance}
\label{subsec:nr_downlink_system}

Here, we provide numerical results to evaluate the overall system performance of NOMA with each of the proposed component channels based on the considered user scheduling and pairing method in Section \ref{sec:pairing} in  a single-cell downlink network with the cell topology described in Fig. \ref{fig:cellStruc}.  The sector radius from the BS to the cell edge was set to be $D_e=500$ m. We considered $B=4$ resource blocks and $K=40$ users  uniformly distributed over the $120^o$ sector from radius 50 m to the cell edge.  The noise power for each user was the same and set to be $N=N_1=N_2=\cdots=N_K=1$. The channel gain $h_{0k}^{(b)}$ from the BS to user $k$ at the resource block $b$ was modelled as the product of a Rayleigh fading factor $f_{0k}^{(b)} \stackrel{i.i.d.}{\sim} \Cc\Nc(0,1)$ and the path loss, given by
\begin{equation}
h_{0k}^{(b)} = f_{0k}^{(b)} \cdot\left(\frac{d_{0k}}{D_e}\right)^{-\gamma},
\end{equation}
where $d_{0k}$ was the distance from the BS to user $k$ and the path loss factor was $\gamma=3$.
The BS transmit power $P_0^{(b)}$ was set so that the expected received SNR at the cell edge was 10 dB, i.e.,
\[
10 \mbox{dB} = \frac{\Ebb\{|h_{0k}^{(b)}|^2\}P_0^{(b)}}{N} = \frac{\Ebb\{|f_{0k}^{(b)}|^2\} \left( \frac{D_e}{D_e} \right)^{-3} P_0^{(b)}}{N} = \frac{P_0^{(b)}}{N} ~~\forall~b=1,\cdots,B.
\]
Thus, users with $d_{0k} < D_e$ had expected SNR larger than 10 dB.  The transmit power $P_1^{(b)}$ of the relaying receiver was set relative to $P_0^{(b)}$. For one realization of user locations, we ran the user scheduling and pairing method in Section \ref{sec:pairing} with the PF forgetting factor  $\tau=0.01$ in \eqref{eq:PFforgetting} for 1000 scheduling intervals, and computed the sum rate divided by 1000 for each scheme. We averaged the sum rate over  50 independent realizations for user locations.   Fig. \ref{fig:total_throughput} shows the sum rate result for NOMA equipped with four different component channels: GBC (simple NOMA), RBC-DF, RBC-CF/NNC, and RBC-CF/NNC/DPC. For the solid lines the near-far paring was used and for the dashed lines  the nearest neighbor pairing was used.
\begin{figure}[tp]
    \centerline{
\scalefig{0.7}
\epsfbox{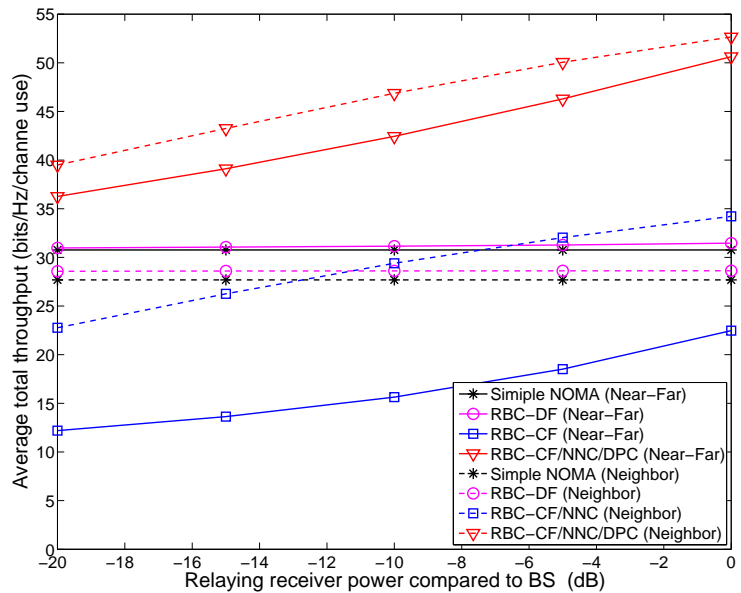}}
    \caption{Total system sum rate: solid line - the near-far paring and dashed line - the nearest neighbor pairing}
    \label{fig:total_throughput}
\end{figure}
It is seen that the gain by RBC-DF is marginal in this operating SNR range with the cell-edge user SNR of 10 dB, as expected from Section \ref{subsec:nr_component_channel}. It is seen that the gain of RBC-CF/NNC/DPC over the simple NOMA is significant when $P_1^{(b)}$ is comparable to $P_0^{(b)}$, as expected from Section \ref{subsec:nr_component_channel}. If the operating SNR is decreased. then the gain of NOMA based on RBC will increase further, as expected from Fig. \ref{fig:rateRegion}(b), (d), and (f). Note that the performance difference due to the two disparate user pairing methods is not so significant for GBC (simple NOMA) and RBC-DF.

\section{Conclusion}
\label{sec:conclusion}

In this paper, we have considered enhancing NOMA by using RBC
component channels in  SISO cellular downlink systems. We have
newly derived an achievable rate region of a RBC with CF/NNC and
have investigated the achievable rate region of a RBC with DF,
CF/NNC, and CF/NNC plus DPC. Based on the achievable rate
analysis, we have investigated the overall system performance of
NOMA equipped with RBC component channels, and have shown that
NOMA with  RBC-DF yields marginal gain and NOMA with
RBC-CF/NNC/DPC yields drastic gain over the simple NOMA based on
GBC in a practical system setup. The gist of the gain of NOMA lies
in non-linear processing to cope with system overloading. By going
beyond simple GBC/SIC to advanced multi-terminal encoding
including DPC and CF/NNC, we can obtain far larger gains.
Currently, active research is going on to implement practical DPC
and CF codes already with some available codes
\cite{Erez&Brink:05IT,Sun&Yang&Liveris&Stankovic&Xiong:09IT,Bennatan&Burshtein&Caire&Shamai:06IT,
Serrano&Thobaben&Rathi&Skoglund:10Asilomar,Serrano:10Phd,Serrano&Thobaben&Andersson&Rathi&Skoglund:12TCOM,Karzand:12Zurich,Karas&Pappi&Karagiannidis:15ComLetter}.
With reflecting the gain in NOMA by using such multi-terminal
encoding, it is worth considering such advanced multi-terminal
encoding for NOMA.

\appendices
\section{Proof of Theorem \ref{prop:rbccf}}
\label{subsec:appendA}

\textit{Codebook Generation:} Fix
$p(x_1)p(u)p(v)p(x_0|u,v)p(\hat{y}_1|y_1,x_1)$. We assume
blockwise\footnote{The term 'block' in the appendix is not the
resource block in the main content of the paper. A  block in
this appendix is a concatenation of $n$ channel code symbols.}
transmission with $n$ code symbols as one block, and transmit $J$
blocks. We randomly and independently generate a codebook for each
block. For each block $j\in[1:J]\defeq\{1,2,\cdots,J\}$,
\begin{itemize}
\item randomly and independently generate $2^{n\hat{R}_2}$ sequences $\xbf_{1j}(l_{j-1})$, $l_{j-1}\in[1:2^{n\hat{R}_2}]$, each according to the distribution $\prod_{k=1}^n p_{X_1}(x_{1,(j-1)n+k})$;
\item  randomly and independently generate $2^{nR_1}$ sequences $\ubf_{j}(m_{1j})$, $m_{1j}\in[1:2^{nR_1}]$, each according to $\prod_{i=1}^n p_{U}(u_{(j-1)n+i})$;
\item randomly and independently generate $2^{nJR_2}$ sequences $\vbf_{1j}(m_2)$, $m_2\in[1:2^{nJR_2}]$, each according to the distribution $\prod_{k=1}^n p_{V}(v_{(j-1)n+k})$;
\item for each $\ubf_j(m_{1j})$ and $\vbf_{j}(m_2)$, randomly generate a sequence $\xbf_{0j}(m_{1j},m_2)$ each according to $\prod_{i=1}^n p_{X|U,V}(x_{0,(j-1)n+i}|u_{(j-1)n+i}(m_{1j}),v_{(j-1)n+i}(m_2))$; and
\item for each $\xbf_{1j}(l_{j-1})$, randomly and conditionally independently generate $2^{n\hat{R}_2}$ sequences $\hat{\ybf}_{1j}(l_j|l_{j-1})$, $l_j\in[1:2^{n\hat{R}_2}]$, each according to $\prod_{i=1}^n p_{\hat{Y}_1|X_1}(\hat{y}_{1,(j-1)n+i}|x_{1,(j-1)n+i}(l_{j-1}))$.
\end{itemize}
Then, the codebook is shared for all nodes. The Markov chain relationship between the codewords ($\xbf_{1j}$, $\ubf_j$, $\vbf_j$, $\xbf_{0j}$, and $\hat{\ybf}_{1j}$) and the received signal vectors ($\ybf_{1j}$ and $\ybf_{2j}$) is described in Fig. \ref{fig:codeMarkov_Chain}.

\textit{Encoding:} Let $m_{1j}$ and $m_2$ be the messages to be sent, and choose $l_0=1$ by convention. The transmitter sends $\xbf_{0j}(m_{1j},m_2)$ generated from $\ubf_j(m_{1j})$ and $\vbf_{j}(m_2)$.

Upon reception of $\ybf_{1j}$, the relaying receiver finds an index $l_j$ such that
\begin{equation}
(\hat{\ybf}_{1j}(l_j|l_{j-1}),\ybf_{1j},\xbf_{1j}(l_{j-1}))\in\mathcal{T}_{\epsilon_1}^{(n)}(\hat{Y}_1,Y_1,X_1),
\end{equation}
where $\mathcal{T}_{\epsilon_1}^{(n)}(\hat{Y}_1,Y_1,X_1)$ is the set of $\epsilon_1-$jointly typical sequences.
If there are more than one such index, choose one of them arbitrarily. If there is no such index, choose an arbitrary index. By the covering lemma \cite{Gamal&Kim:11NIT}, if $\hat{R}_2>I(\hat{Y}_1;Y_1|X_1)+\delta_1(\epsilon_1)$, the probability that there exists at least one such index tends to $1$ as $n\rightarrow\infty$, where $\epsilon_1>0$ and $\delta_1(\cdot)$ is a positive function such that $\delta_1(\epsilon_1)\rightarrow0$ as $\epsilon_1\rightarrow0$.  After determining $l_j$, the relaying receiver transmits $\xbf_{1,j+1}(l_j)$ at the next block $j+1$.

\textit{Decoding at the Relaying Receiver:} At the end of each block $j$, the relaying receiver finds the unique message $\hat{m}_{1j}\in[1:2^{nR_1}]$ such that
\begin{equation}
(\ubf_{j}(\hat{m}_{1j}),\ybf_{1j})\in\mathcal{T}_{\epsilon_2}^{(n)}(U,Y_1),
\end{equation}
where $\epsilon_2>\epsilon_1$. If there are no or more than one such messages,  declare  error.

\textit{Decoding at the Second Receiver:} At the end of the whole transmission of $J$ blocks, the second receiver finds the unique message $\hat{m}_2\in[1:2^{nJR_2}]$ such that
\begin{equation}
(\vbf_{j}(\hat{m}_2),\xbf_{1j}(\hat{l}_{j-1}),\hat{\ybf}_{1j}(\hat{l}_j|\hat{l}_{j-1}),\ybf_{2j})\in\mathcal{T}_{\epsilon_3}^{(n)}(V,X_1,\hat{Y}_1,
Y_2)
\end{equation}
for all $j\in[1:J]$ for some $\hat{l}_1,\hat{l}_2,\ldots,\hat{l}_J$, where $\epsilon_3>\epsilon_1$. If there are no or more than one such messages,  declare  error.

\begin{figure*}[ht]
\centerline{
\leavevmode
\strut\AffixLabels{
\psfrag{n0}[h][h]{\small Transmitter$~~$} %
\psfrag{n11}[h][h]{\small Relaying$~~$} %
\psfrag{n12}[h][h]{\small Receiver$~~$} %
\psfrag{n21}[h][h]{\small Second$~~$} %
\psfrag{n22}[h][h]{\small Receiver$~~$} %
\psfrag{x1}[h][h]{\small $\xbf_{1j}$} %
\psfrag{u}[h][h]{\small $\ubf_j$} %
\psfrag{v}[h][h]{\small $\vbf_j$} %
\psfrag{x0}[h][h]{\small $\xbf_{0j}$} %
\psfrag{yh}[h][h]{\small $\widehat{\ybf}_{1j}$} %
\psfrag{y1}[h][h]{\small $\ybf_{1j}$} %
\psfrag{y2}[h][h]{\small $\ybf_{2j}$} %
\psfrag{p}[l]{\small possible self-interference} %
\psfrag{sup}[c]{\small \tcr{Superposition/DPC}} %
\psfrag{cf}[l]{\small \tcb{CF/NNC}} %
\scalefig{0.5}\epsfbox{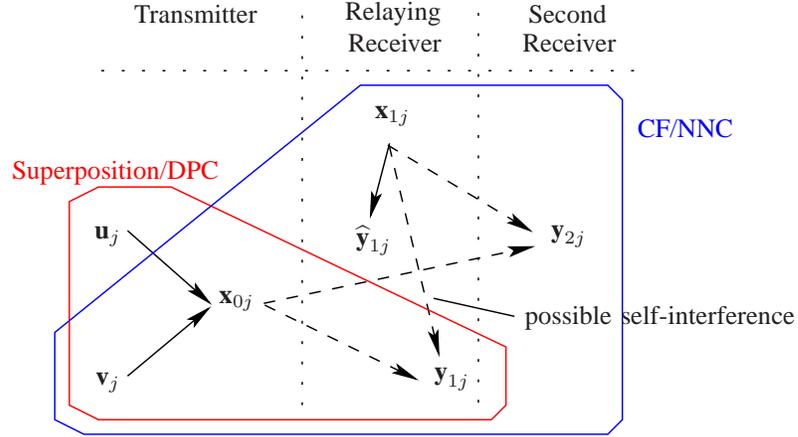} } }
\caption{Markov chain relationship between codewords (solid arrows: codeword Markov chain and dashed arrows: channel links) } \label{fig:codeMarkov_Chain}
\end{figure*}

\textit{Analysis of the Error Probability:} Without loss of generality, we assume that truly transmitted message indices are $M_{11}=\cdots=M_{1J}=M_2=1$ and $L_1=\cdots=L_J=1$. Then, decoding error occurs  only if one or more of the following events occur:
\begin{itemize}
\item $\Ec_1 := \{(\hat{\Ybf}_{1j}(l_j|1),\Xbf_{1j}(1),\Ybf_{1j})\notin\mathcal{T}_{\epsilon_1}^{(n)}\mbox{ for all }l_j \mbox{ for some }j\in[1:J]\}$.
\item $\Ec_2 :=\{(\Ubf_j(1),\Ybf_{1j})\not\in\mathcal{T}_{\epsilon_2}^{(n)}$ for some $j\in[1:J]\}$.
\item $\Ec_3 := \{(\Ubf_j(m_{1j}),\Ybf_{1j})\in\mathcal{T}_{\epsilon_2}^{(n)}$ for some $m_{1j}\ne1$ and for some $j\in[1:J]\}$.
\item $\Ec_4 := \{(\Vbf_{j}(1),\Xbf_{1j}(1),\hat{\Ybf}_{1j}(1|1),\Ybf_{2j})\notin\mathcal{T}_{\epsilon_3}^{(n)}\mbox{ for some }j\in[1:J]\}$.
\item $\Ec_5 := \{(\Vbf_{j}(m_2),\Xbf_{1j}(l_{j-1}),\hat{\Ybf}_{1j}(l_j|l_{j-1}),\Ybf_{2j})\in\mathcal{T}_{\epsilon_3}^{(n)}\mbox{ for all $j$ for some } (l_1,\cdots,l_J),~m_2\ne 1\}$,
\end{itemize}
where the notations for typical sets are simplified. By the union bound, the error probability is bound as follows:
\begin{equation}
P(\Ec)\le P(\Ec_1)+P(\Ec_2\cap\Ec_1^c)+P(\Ec_3\cap\Ec_1^c)+P(\Ec_4\cap\Ec_1^c)+P(\Ec_5).
\end{equation}
The first term $P(\Ec_1)$ tends to zero as $n\rightarrow\infty$ by the covering lemma \cite{Gamal&Kim:11NIT} if $\hat{R}_2>I(\hat{Y}_1;Y_1|X_1)+\delta_1(\epsilon_1)$. The second term $P(\Ec_2)$ tends to zero as $n\rightarrow\infty$  because $\Ubf_j(1)\rightarrow\Ybf_{1j}$. The third term $P(\Ec_3)$ tends to zero as $n\rightarrow\infty$ by the packing lemma \cite{Gamal&Kim:11NIT} if
\begin{equation} \label{eq:theorem3R1}
R_1<I(U;Y_1).
\end{equation}
 The fourth term $P(\Ec_4\cap\Ec_1^c)$ tends to zero as $n\rightarrow\infty$ by the Markov lemma \cite{Gamal&Kim:11NIT}, since $(\Vbf_{j}(1),\Xbf_{1j}(1),$ $\hat{\Ybf}_{1j}(1|1))\in\mathcal{T}_{\epsilon_1}^{(n)}$ and
\begin{equation}
\hat{\Ybf}_{1j}\rightarrow(\Vbf_{j},\Xbf_{1j})\rightarrow\Ybf_{2j}.
\end{equation}
Finally, for the fifth term (The proof written in here is similar to that of \cite{Lim&Kim&Gamal&Chung:11IT}), define the events
\begin{equation}
\tilde{\Ec}_j(m,l_{j-1},l_j)=\{(\Vbf_{j}(m),\Xbf_{1j}(l_{j-1}),\hat{\Ybf}_{1j}(l_j|l_{j-1}), \Ybf_{2j})\in\mathcal{T}_{\epsilon}^{(n)}\}.\label{eq:event}
\end{equation}
Then, we can see that
\begin{align}
P(\Ec_5) &= P(\cup_{m\ne 1}\cup_{l_1,\cdots,l_J}\cap_{j=1}^J \tilde{\Ec}_j(m,l_{j-1},l_j))\\
         &\le \sum_{m\ne 1}\sum_{l_1,\cdots,l_J}P(\cap_{j=1}^J \tilde{\Ec}_j(m,l_{j-1},l_j))\\
         &= \sum_{m\ne 1}\sum_{l_1,\cdots,l_J}\prod_{j=1}^J P(\tilde{\Ec}_j(m,l_{j-1},l_j))\\
         &\le \sum_{m\ne 1}\sum_{l_1,\cdots,l_J}\prod_{j=2}^J P(\tilde{\Ec}_j(m,l_{j-1},l_j)).  \label{eq:append42}
\end{align}
Now, consider the probability of the event \eqref{eq:event}. First, assume that $l_{j-1}=1$. Then, by the joint typicality lemma \cite{Gamal&Kim:11NIT} we have for $l_{j-1}=1$,
\begin{align}
P(\tilde{\Ec}_j(m,l_{j-1},l_j)) &=P\{(\Vbf_{j}(m_2),\Xbf_{1j}(l_{j-1}),\hat{\Ybf}_{1j}(l_j|l_{j-1}),\Ybf_{2j})\in\mathcal{T}_{\epsilon_3}^{(n)}\}\\
                                &\le 2^{-n(I_1-\delta_3(\epsilon_3))}, \label{eq:l=1}
\end{align}
where $I_1 = I(V;\hat{Y}_1,Y_2|X_1)$, since $\Vbf_{j}(m_2)$
is independent of $\hat{\Ybf}_{1j}(l_j|l_{j-1})$ and $\Ybf_{2j}$ for given $\Xbf_{1j}(l_{j-1})$ due to $M_2=1\ne m_2$. Second, assume that $l_{j-1}\ne 1$. Then,  $(\Vbf_{j}(m_2),\Xbf_{1j}(l_{j-1}),\hat{\Ybf}_{1j}(l_j|l_{j-1}))$ is independent of $\Ybf_{2j}$.  Then, by \cite[Lemma 2]{Lim&Kim&Gamal&Chung:11IT}, which is an application of the joint typicality lemma, we have for $l_{j-1}\ne 1$,
\begin{equation}
P(\tilde{\Ec}_j(m,l_{j-1},l_j))\le 2^{-n(I_2-\delta_3(\epsilon_3))}, \label{eq:lne1}
\end{equation}
where $I_2=I(V,X_1;Y_2)+I(\hat{Y}_1;V,Y_2|X_1)$. If $l_1,l_2,\cdots,l_{J-1}$ have  $k$ $1$'s, then by \eqref{eq:l=1} and \eqref{eq:lne1} we have
\begin{equation}
\prod_{j=2}^n P(\tilde{\Ec}_j(m,l_{j-1},l_j))\le 2^{-n(kI_1+(J-1-k)I_2-(J-1)\delta_3(\epsilon_3))}.
\end{equation}
Therefore, from \eqref{eq:append42} we have
\begin{align}
P(\Ec_5) &\le \sum_{m\ne 1}\sum_{l_1,\cdots,l_J}\prod_{j=2}^J P(\tilde{\Ec}_j(m,l_{j-1},l_j))\\
         &\le \sum_{m\ne 1}\sum_{l_J}\sum_{l_1,\cdots,l_{J-1}}\prod_{j=2}^b P(\tilde{\Ec}_j(m,l_{j-1},l_j))\\
         &\le \sum_{m\ne 1}\sum_{l_J}\sum_{k=0}^{J-1}\left(\begin{array}{c}J-1\\k\end{array}\right) 2^{n(J-1-k)\hat{R}_2}\cdot 2^{-n(kI_1+(J-1-k)I_2-(J-1)\delta_3(\epsilon_3))}  \label{eq:append49}\\
         &= \sum_{m\ne 1}\sum_{l_J}\sum_{k=0}^{J-1}\left(\begin{array}{c}J-1\\k\end{array}\right) 2^{-n(kI_1+(J-1-k)(I_2-\hat{R}_2)-(J-1)\delta_3(\epsilon_3))}\\
         &\le 2^{nJR_2}\cdot 2^{n\hat{R}_2}\cdot 2^J\cdot 2^{-n((J-1)\min\{I_1,I_2-\hat{R}_2\}-(J-1)\delta_3(\epsilon_3))},
\end{align}
which tends to zero as $n\rightarrow\infty$, if
\begin{equation}
R_2<\frac{J-1}{J}(\min\{I_1,I_2-\hat{R}_2\}-\delta_3(\epsilon_3))-\frac{\hat{R}_2}{J}.
\end{equation}
(In \eqref{eq:append49}, the term
$2^{n(J-1-k)\hat{R}_2}$ accounts for the number of $l_{j-1}\ne 1$.
Eliminating $\hat{R}_2$ by substituting $I(\hat{Y}_1;Y_1|X_1)+\delta_1(\epsilon_1)$ from the condition $\hat{R}_2>I(\hat{Y}_1;Y_1|X_1)+\delta_1(\epsilon_1)$ and sending $J\rightarrow \infty$, we obtain
\begin{equation} \label{eq:theorem3R2}
R_2<\min\{I(V;\hat{Y}_1,Y_2|X_1), I(V,X_1;Y_2)-I(\hat{Y}_1;Y_1|V,X_1,Y_2)\}-\delta_1(\epsilon_1)-\delta_3(\epsilon_3).
\end{equation}
Since $\delta_1$ and $\delta_3$ converge to zero, we have the claim by \eqref{eq:theorem3R1} and \eqref{eq:theorem3R2}.
\hfill{$\blacksquare$}

\section{Achievable Rate Region for the RBC-CF/NNC scheme in the Gaussian Case}
\label{subsec:appendB}

In the Gaussian case, we have $p(u) \sim \Cc\Nc(0, \alpha P_0)$,  $p(v) \sim \Cc\Nc(0,\bar{\alpha}P_0)$, and $p(x_1) \sim \Cc\Nc(0, P_1)$. Furthermore, we have  \eqref{eq:GaussianSuperpose} and  \eqref{eq:GaussianCFhatY}  for $p(x_0|u,v)$ and $p(\hat{y}_1|y_1,x_1)$, respectively.
  We need to compute $R_1$ and $R_2$ in \eqref{eq:RBC_CF_Achievable2} and \eqref{eq:RBC_CF_Achievable} based on \eqref{eq:GaussianSuperpose},   \eqref{eq:GaussianCFhatY},
\eqref{eq:GaussianChannel1}, and \eqref{eq:GaussianChannel2}.
Since
\begin{align}
Y_1 &= h_{01}(U+V) + Z_1\\
\hat{Y}_1 &= h_{01}V + Z_1 + \hat{Z}\\
Y_2 &= h_{02}(U+V) + h_{12}X_1 + Z_2,
\end{align}
the achievable rate region in Theorem \ref{prop:rbccf} is given by
\begin{align}
R_1 <& I(U;Y_1)\nonumber\\
    =& I(U;h_{01}U + h_{01}V + Z_1) \label{eq:Gaussian_RBC_CF_R1}\\
R_2 <& \min\{I(V;\hat{Y}_1,Y_2|X_1),I(V,X_1;Y_2)-I(\hat{Y}_1;Y_1|V,X_1,Y_2)\}\nonumber\\
    =& \min\{I(V;h_{01}V+Z_1+\hat{Z},h_{02}V + h_{02}U + Z_2),\nonumber\\
     & I(V,X_1;h_{02}V+h_{12}X_1 + h_{02}U+Z_2)-I(Z_1+\hat{Z};h_{01}U+Z_1|h_{02}U+Z_2)\}\label{eq:Gaussian_RBC_CF_R2}
\end{align}
Then, the term in \eqref{eq:Gaussian_RBC_CF_R1} and the first argument of the minimum in \eqref{eq:Gaussian_RBC_CF_R2} are respectively given by
\begin{align}
I(U;h_{01}U + h_{01}V + Z_1) &= \log\left(1+\frac{|h_{01}|^2\alpha P_0}{|h_{01}|^2\bar{\alpha}P_0 + N_1}\right)\\
I(V;h_{01}V+Z_1+\hat{Z},h_{02}V + h_{02}U + Z_2) &= \log\left(1+\frac{|h_{01}|^2\bar{\alpha}P_0}{N_1+\hat{N}}+\frac{|h_{02}|^2\bar{\alpha}P_0}{|h_{02}|^2\alpha P_0+N_2}\right)
\end{align}
The first term  of the second argument in the minimum in \eqref{eq:Gaussian_RBC_CF_R2} is expressed as
\begin{align}
I(V,X_1;h_{02}V+h_{12}X_1 + h_{02}U+Z_2) &= \log\left(1+\frac{|h_{02}|^2\bar{\alpha}P_0+|h_{12}|^2P_1}{|h_{02}|^2{\alpha}P_0+N_2}\right).
\end{align}
Finally, the second term of the second argument in the minimum in \eqref{eq:Gaussian_RBC_CF_R2} can be expressed as
\begin{align}
&I(Z_1+\hat{Z};h_{01}U+Z_1|h_{02}U+Z_2)\nonumber\\
&= h(Z_1+\hat{Z}|h_{02}U+Z_2)-h(Z_1+\hat{Z}|h_{01}U+Z_1,h_{02}U+Z_2) \nonumber\\
&= h(Z_1+\hat{Z})-h(Z_1+\hat{Z}|h_{01}U+Z_1,h_{02}U+Z_2)\nonumber\\
&= \log\left(N_1 + \hat{N}\right)-\log\left(\frac{N_2|h_{01}|^2\alpha P_0 + N_1|h_{02}|^2\alpha P_0 + N_1N_2}{N_1N_2|h_{01}|^2\alpha P_0 + \hat{N}N_2|h_{01}|^2\alpha P_0 + \hat{N}N_1|h_{02}|^2\alpha P_0 +\hat{N}N_1N_2}\right)\nonumber\\
&= \log\left(1+\frac{N_1^2N_2+N_1^2|h_{02}|^2\alpha P_0}{\hat{N}N_1N_2+\hat{N}N_2|h_{01}|^2\alpha P_0+\hat{N}N_1|h_{02}|^2\alpha P_0+N_1N_2|h_{01}|^2\alpha P_0}\right).
\end{align}



\begin{thebibliography}{}


\bibitem{Li&Niu&Papathanassiou&Wu:14VT} Q. Li, H. Niu, A. Papathanassiou, and G. Wu, ``5G Network Capacity: Key Elements and Technologies,"
{\it IEEE Veh. Technol. Mag.}, vol. 9, pp. 71 -- 78, Mar. 2014.

\bibitem{Saito&Kishiyama&Benjebbour&Nakamura&Li&Higuchi:13VTC} Y. Saito, Y. Kishiyama, A. Benjebbour, T. Nakamura, A. Li, and K. Higuchi, ``Non-Orthogonal Multiple Access (NOMA) for Cellular Future Radio Access," in {\it Proc. IEEE VTC Spring}, 2013.

\bibitem{Ding&Peng&Poor:15ComLetter} Z. Ding, M. Peng, and H. V. Poor, ``Cooperative Non-Orthogonal Multiple Access in 5G Systems," {\it IEEE Commun. Letter}, vol. 19, pp. 1462 -- 1465, Aug. 2015.

\bibitem{Cover&Thomas:91} T. M. Cover and J. A. Thomas, {\it Elements of Information Theory.} New York: Wiley, 1991.

\bibitem{Gamal&Kim:11NIT} A. E. Gamal and Y.-H. Kim, {\it Network Information Theory.} New York: Cambridge University Press, 2011.

\bibitem{Doppler&Rinne&Wijting&Ribeiro&Hugl:09ComMag}  K. Doppler, M. Rinne, C. Wijting, C. B. Ribeiro, and K. Hugl, ``Device-to-Device Communication as an Underlay to LTE-Advanced Networks," {\it IEEE Communn. Mag.}, vol. 7, pp. 42 -- 49, Dec. 2009.

\bibitem{Fodor&Dahlman&Mildh&Parkvall&Reider&Miklos&Turanyi:12ComMag} G. Fodor, E. Dahlman, G. Mildh, S. parkvall, N. Reider, G. Mikl$\acute{\mbox{o}}$s, and Z. Tur$\acute{\mbox{a}}$nyi, ``Design Aspects of Network Assisted Device-to-Device Communications," {\it IEEE Communn. Mag.}, vol. 50, pp. 170 -- 177, Mar. 2012.

\bibitem{Liang&Veeravalli:04ISIT} Y. Liang and V. V. Veeravalli, ``The Impact of Relaying on the Capacity of Broadcast Channels," in {\it Information Theory Proceedings (ISIT), 2004 IEEE International Symposium on}, (Chicago, USA), pp. 403 -- 403, Jun. 2004.

\bibitem{Bross:09IT} S. I. Bross, ``On the Discrete Momoryless Partially Cooperatve Relay Broadcast Channel and the Broadcast Channel with Cooperative Decoders," {\it IEEE Trans. Inform. Theory}, vol. 55, pp. 2161 -- 2182, May. 2004.

\bibitem{Cover&Gamal:79IT} T. M. Cover and A. E. Gamal, ``Capacity Theorems for the Relay Channel,"
{\it IEEE Trans. Inform. Theory}, vol. 25, pp. 572 –- 584, Sep. 1979.

\bibitem{Kramer&Gastpar&Gupta:05IT} G. Kramer, M. Gastpar, and P. Gupta, ``Cooperative Strategies and Capacity Theorems for Relay Networks," {\it IEEE Trans. Inform. Theory}, vol. 51, pp. 3037 -- 3063, Sep. 2005.

\bibitem{Lim&Kim&Gamal&Chung:11IT} S. H. Lim, Y.-H. Kim, A. E. Gamal, and S.-Y. Chung, ``Noisy Network Coding," {\it IEEE Trans. Inform. Theory}, vol. 57, pp. 3132 -- 3152, May. 2011.

\bibitem{Costa:83IT} M. H. M. Costa, ``Writing on Dirty paper," in {\it IEEE Trans. Inform. Theory}, vol. 29, pp. 439 -- 441, May. 1979.

\bibitem{Gelfand&Pinsker:80PCIT} S. I. Gel'fand and M. S. Pinsker, ``Coding for Channel with Random Paramters," {\it Probl. Contr. Inform. Theory}, vol. 9, pp. 19 -- 31, Jan. 1980.

\bibitem{Benjebbour&Li&Saito&Kishiyama&Harada&Nakamura:13Globecom} A. Benjebbour, A. Li, Y. Saito, Y. Kishiyama, A. Harada, and T. Nakamura, ``System-Level Performance of Downlink NOMA for Future LTE Enhancements," in {\it Proc. IEEE Globecom 2013}, (Atlanta, USA), Dec. 2013.

\bibitem{Viswanath&Tse&Laroia:02IT} P. Viswanath, D. Tse, and R. Laroia, ``Opportunistic Beamforming Using Dumb Antennas," {\it IEEE Trans. Inform. Theory}, vol. 48, pp. 1277 -- 1294, Jun. 2002.

\bibitem{Holma&Toskala:book} H. Holma and A. Toskala, {\it WCDMA for UMTS.} New York: Wiley, 2001.

\bibitem{Erez&Brink:05IT} U. Erez and S. ten Brink, ``A Close-to-Capacity Dirty Paper Coding Scheme," {\it IEEE Trans. Inform. Theory}, vol. 51, pp. 3417 -- 3432, Oct. 2005.

\bibitem{Sun&Yang&Liveris&Stankovic&Xiong:09IT} Y. Sun, Y. Yang, A. D. Liveris, V. Stankovi$\acute{\mbox{c}}$, and Z. Xiong, ``Near-Capacity Dirty-Paper Code Design: A Source-Channel Coding Approach," {\it IEEE Trans. Inform. Theory}, vol. 55, pp. 3013 -- 3031, Jul. 2009.

\bibitem{Bennatan&Burshtein&Caire&Shamai:06IT} A. Bennatan, D. Burshtein, G. Caire, and S. Shamai, ``Superposition Coding for Side-Information Channels," {\it IEEE Trans. Inform. Theory}, vol. 52, pp. 1872 -- 1889, May. 2006.

\bibitem{Serrano&Thobaben&Rathi&Skoglund:10Asilomar} R. B.-Serrano, R. Thobaben, V. Rathi, and M. Skoglund, ``Polar Codes for Compress-and-Forward in Binary Relay Channels," in {\it Proc. Asilomar Conf. Signals, Systems and Computers, Pacific Grove, CA. USA.} Nov. 2010.

\bibitem{Serrano:10Phd} R. B.-Serrano, ``Coding Strategies for Compress-and-Forward Relaying," {\it Licentiate Thesis, Royal Institutue of Technology (KTH), Stockholm, Sweden}, 2010.

\bibitem{Serrano&Thobaben&Andersson&Rathi&Skoglund:12TCOM} R. B.-Serrano, R. Thobaben, M. Andersson, V. Rathi, and M. Skoglund, ``Polar Codes for Cooperative Relaying," {\it IEEE Trans. Commun.}, vol. 60, pp. 3263 -- 3273, Nov. 2012.

\bibitem{Karzand:12Zurich} M. Karzand, ``Polar Codes for Degraded Relay Channels," in {\it Proc. Int. Zurich Seminar Commun., Zurich, Switzerland}, Feb. 2012.

\bibitem{Karas&Pappi&Karagiannidis:15ComLetter} D. S. Karas, K. N. Pappi, and G. K. Karagiannidis, ``Smart Decode-and-Forward Relaying with Polar Codes," {\it IEEE Commun. Letter}, vol. 3, pp. 62 -- 65, Feb. 2014.

\end{thebibliography}
\end{document}